\documentclass[11pt,letterpaper]{article}
\usepackage{amsmath,amsfonts,amssymb,amsthm}
\usepackage{fullpage}
\usepackage{latexsym}
\usepackage{xspace}
\usepackage{paralist}
\usepackage{xy}
\usepackage{url}
\usepackage{tabularx}
\usepackage{color}
\usepackage[bookmarksdepth=2]{hyperref}

\input xy
\xyoption{all}
\newtheorem{theorem}{Theorem}
\newtheorem*{theorem*}{Theorem}
\newtheorem{corollary}[theorem]{Corollary}
\newtheorem{question}{Question}
\newtheorem{lemma}{Lemma}
\newtheorem{observation}{Observation}
\newtheorem{proposition}[theorem]{Proposition}
\newtheorem*{proposition*}{Proposition}

\newtheorem{criterion}{Criterion}

\theoremstyle{definition}
\newtheorem{problem}[question]{Problem}
\newtheorem{example}{Example}
\newtheorem{remark}{Remark}

\newcommand{\Iff}{iff\xspace}
\newcommand{\Sec}[1]{Section~\ref{#1}\xspace}

\newcommand{\N}{\mathbb{N}}
\newcommand{\F}{\mathbb{F}}

\newcommand{\Z}{\mathbb{Z}}

\newcommand{\ie}{i.\,e.}
\newcommand{\eg}{e.\,g.}

\renewcommand{\P}{{\rm P}\xspace}
\newcommand{\NP}{{\rm NP}\xspace}

\newcommand{\cc}[1]{\mathsf{#1}} 

\newcommand{\algprob}[1]{{\sc #1}\xspace}
\newcommand{\algprobm}[1]{\mbox{\sc #1}\xspace}
\newcommand{\GrI}{\algprobm{GraphI}}
\newcommand{\GpI}{\algprobm{GpI}}

\newcommand{\cohiso}{{\sc Cohomology Class Isomorphism}\xspace}
\newcommand{\actcomp}{{\sc Action Compatibility}\xspace}

\newcommand{\edpcs}{{\sc EDPC}\xspace}
\newcommand{\cohisos}{{\sc CCIso}\xspace}
\newcommand{\actcomps}{{\sc ActComp}\xspace}

\newcommand{\aut}{\mathrm{Aut}} 
\newcommand{\sg}{\mathrm{S}}  
\newcommand{\Ind}{\mathrm{Ind}}
\newcommand{\clo}{\mathrm{Clo}}
\newcommand{\Iota}{I}
\newcommand{\id}{\mathrm{id}} 

\newcommand{\Iso}{\mathrm{Iso}}
\newcommand{\Aut}{\mathrm{Aut}}

\newcommand{\Gal}{\mathrm{Gal}}

\newcommand{\PSL}{\mathrm{PSL}}
\newcommand{\PSU}{\mathrm{PSU}}
\newcommand{\PGL}{\mathrm{PGL}}
\newcommand{\SL}{\mathrm{SL}}

\DeclareMathOperator{\Char}{char}
\DeclareMathOperator{\im}{Im}

\newcommand{\GL}{\mathrm{GL}}
\newcommand{\M}{\mathrm{M}}

\newcommand{\poly}{\mathrm{poly}}
\newcommand{\sym}{\mathrm{Sym}}

\DeclareMathOperator{\rad}{Rad}

\newcommand{\dihedral}[1]{\mathrm{D}_{2^{#1}}}
\newcommand{\sdihedral}[1]{\mathrm{SD}_{2^{#1}}}
\newcommand{\gquaternion}[1]{\mathrm{GQ}_{2^{#1}}}

\newcommand{\cohisol}{{\sc Cohomology Class Isomorphism}\xspace}
\newcommand{\actcompl}{{\sc Action Compatibility}\xspace}
\newcommand{\edpcl}{{\sc Extension Data Pseudo-congruence}\xspace}

\renewcommand{\cohiso}{{\sc CCIso}\xspace}
\renewcommand{\actcomp}{{\sc ActComp}\xspace}

\renewcommand{\Iff}{if and only if\xspace}
\newcommand{\wrt}{with respect to\xspace}
\newcommand{\st}{such that\xspace}

\newtheorem*{lemma*}{Lemma}

\title{Polynomial-time isomorphism test of groups that are tame extensions}

\author{
Joshua A. Grochow\thanks{Santa Fe Institute, Santa Fe, NM, USA, 
{\tt jgrochow@santafe.edu}} 
\and 
Youming Qiao\thanks{Centre for Quantum 
Computation and Intelligent Systems, University of Technology, Sydney, Australia, 
{\tt jimmyqiao86@gmail.com}}}

\date{\today}

\begin{document}

\maketitle

\begin{abstract}
We give new polynomial-time algorithms for testing isomorphism of a class of groups given by multiplication tables (\GpI). 
Two results (Cannon \& Holt, J. Symb. Comput. 2003; Babai, Codenotti \& 
Qiao, ICALP 2012) imply that \GpI reduces to the 
following: given groups $G, H$ with characteristic subgroups of the same 
type and
isomorphic to $\Z_p^d$, and given the coset of isomorphisms $\Iso(G/\Z_p^d, 
H/\Z_p^d)$, 
compute $\Iso(G, H)$ in time $\poly(|G|)$. Babai \& Qiao (STACS 2012) solved this problem  
when a Sylow $p$-subgroup of $G/\Z_p^d$ is trivial.
In this paper, we solve the preceding problem in the so-called ``tame'' case, \ie,
when a Sylow $p$-subgroup of $G/\Z_p^d$ is cyclic, dihedral,
semi-dihedral, or generalized 
quaternion. These cases correspond exactly to the group algebra 
$\overline{\F}_p[G/\Z_p^d]$ being of 
tame type, as in the celebrated tame-wild dichotomy in representation theory. 
We then solve new cases of \GpI in polynomial time.

Our result relies crucially on the divide-and-conquer strategy proposed earlier by 
the authors (CCC 2014), which splits \GpI into two problems, one on group actions 
(representations), and one on group cohomology. Based on this strategy, we combine 
permutation group and representation 
algorithms with new mathematical results, including bounds on the number of 
indecomposable representations of groups in the tame case, and on the size of 
their cohomology groups. 

Finally, we note that when a group extension is \emph{not} tame, the preceding 
bounds do not hold. This suggests a precise sense in which the tame-wild dichotomy 
from representation theory may also be a dividing line between the (currently) easy and hard 
instances of \GpI.
\end{abstract}

\section{Introduction}
The group isomorphism problem \GpI is to decide whether two finite groups, given by 
their multiplication tables, are isomorphic. It is one of the few natural 
problems not known to be in $\P$, and unlikely to be $\NP$-complete, as it reduces to Graph Isomorphism (\GrI; see, \eg, \cite{struct}).
In addition to its intrinsic interest, resolving the exact complexity of $\GpI$ is thus a tantalizing question. Further, there is
a surprising connection between $\GpI$ and the Geometric Complexity Theory program (see, \eg, \cite{gct_acm} and references therein): Techniques from \GpI were used to solve cases of \algprob{Lie Algebra Isomorphism} that have applications in Geometric Complexity Theory \cite{grochowLie}. In a survey article \cite{Babaisurvey} in 1995, after enumerating several isomorphism-type problems including $\GrI$ and $\GpI$, Babai expressed the belief that \GpI might be the only one expected to be in $\cc{P}$.\footnote{The exact quotation from Babai's 1995 survey \cite{Babaisurvey} is: ``None of the problems mentioned in this section, with the possible exception of isomorphism of groups given by a Cayley table, is expected to have polynomial time solution.''} 
Despite its connection with \GrI, $\cc{P}$ seems an achievable goal for \GpI, as there are many reasons \GpI seems easier than \GrI (see, \eg, the introduction to \cite{GQccc} for an overview of these reasons).

As a group of order $n$ can be generated by $\lceil \log n\rceil$ elements, 
$\GpI$ is solvable in time $n^{\log n+O(1)}$ \cite{FN,Mil78}.\footnote{Miller 
\cite{Mil78} attributes this algorithm to Tarjan.} The only 
improvement for the general case was Rosenbaum's recent $n^{0.5\log 
n+O(1)}$ \cite{Rosen2}. However, there have been more significant improvements for 
special group classes, representing a more structural approach to the problem. 
Isomorphism of Abelian groups was recognized as easy quite early 
\cite{Sav80,Vik96}, 
leading to an $O(n)$-time algorithm \cite{Kav07}. Since 2009, 
there have been several non-trivial polynomial-time algorithms for much more 
complicated group classes: groups with no Abelian normal subgroups 
\cite{BCGQ,BCQ}, groups with Abelian Sylow towers \cite{Gal09,QST11,BQ}, and 
quotients of generalized Heisenberg groups \cite{LW12}. 

Partly motivated to distill 
a common pattern from the three recent major polynomial-time 
algorithms \cite{BCQ,BQ,LW12}, the authors proposed \cite{GQccc} a 
divide-and-conquer strategy for $\GpI$ based on the extension theory of groups. 
This strategy is crucial for Theorem~\ref{thm:main}. Before 
getting to the details of this strategy, let us first 
examine an approach for \GpI that motivates the problem that we study. 

In 2003, Cannon and Holt \cite{CH03} suggested the 
following outline for \GpI. First, they introduce a natural sequence of 
characteristic 
subgroups:
$G=G_0\rhd G_1\rhd \dots \rhd G_\ell=\id, $
where $G_1=\rad(G)$ is the solvable radical of $G$---the largest solvable normal 
subgroup---and $G_i/G_{i+1}$ is elementary Abelian for all $1\leq i\leq 
\ell-1$. 
This filtration is easily computed, and for each 
factor we know how to test isomorphism: $G/\rad(G)$ has no Abelian normal 
subgroups, so is handled by \cite{BCQ}.

Given two groups $G$ and $H$, after computing these filtrations of $G$ and $H$, the strategy is to first test isomorphisms of the corresponding factors, which is necessary for $G$ 
and $H$ to be isomorphic. Then, starting from $G_0/G_1(=G/\rad(G))$, proceed 
inductively along this filtration. Note that for $G_0 / G_1$, not only is isomorphism decidable in polynomial time, but a generating set for the 
coset of isomorphisms $\Iso(G_0 / G_1, H_0 / H_1)$ can be found in 
polynomial time \cite{BCQ}. 
After this initial step, a positive solution to the following 
problem would show that $\GpI \in \P$:

\begin{problem} \label{prob:key}
Given two groups $G, H$ with characteristic elementary Abelian subgroups $A$ 
and $B$, respectively, compute $\Iso(G, H)$ from $\Iso(G/A, H/B)$ in time $\poly(|G|)$.
\end{problem}

In fact, by developing a heuristic algorithm for Problem~\ref{prob:key} in 
\cite[Sec.~5]{CH03}, Cannon \& Holt obtained a practical algorithm for \GpI, but 
their algorithm uses a backtrack search that does not have good worst-case 
guarantees.\footnote{Due to different goals and settings, it is natural that our setting and the setting of Cannon \& Holt use different algorithmic ideas. 
That is, Cannon \& Holt work with more succinct representations of groups, and their goal is to 
obtain algorithms fast in practice, even if only heuristically. We work with the more ``redundant'' Cayley tables, 
but our goal is worst-case analysis.} Still, this is a 
very natural approach, and the polynomial-time algorithm for testing isomorphisms for $G/\rad(G)$ \cite{BCQ} solves the first step to 
this approach in the Cayley table model. 

To the best of our knowledge, the only previous result about 
Problem~\ref{prob:key} 
with a worst-case analysis in the Cayley table model is by Babai and the second 
author \cite{BQ}, who solved the case when $A\cong \Z_p^k$ and the Sylow 
$p$-subgroup\footnote{Although Sylow $p$-subgroups of a group need not be unique, for a given $p$ they are all isomorphic, so we may speak of ``the'' Sylow $p$-subgroup unambiguously, when we only need to refer to its isomorphism type.}
 of $G/A$ is trivial; that is, when $p\nmid |G/A|$. 
This was the key to the main result in \cite{BQ}.

In this paper, we solve Problem~\ref{prob:key} under 
certain conditions on the Sylow subgroups of $G$, more general than the 
aforementioned one for $\cite{BQ}$. Furthermore, these conditions are very 
natural, as they are aligned with the celebrated tame-wild dichotomy in the 
representation theory of associative algebras \cite{CR66,Ben98}. 

The following is a high-level picture of the tame-wild dichotomy; defining tame and wild rigorously requires terminology that is unnecessary for this article; we refer to \cite[Sec.~4.4]{Ben98} for a comprehensive introduction. For an algebra $L$ over an infinite field, classifying its
indecomposable representations up to isomorphism---those representations that are 
not direct sums of smaller ones---is a fundamental problem. The nicest possibility 
is when 
there are only finitely many indecomposables, in which case $L$ is said to be of 
\emph{finite type}. Beyond this, some algebras have the property that their indecomposables come in finitely many one-parameter families in each fixed dimension $d$,\footnote{For readers not familiar 
with this concept, here is an example to illustrate intuitively what one-parameter families mean. For an algebraically 
closed field $\F$, the Jordan blocks form a one-parameter family with the 
eigenvalue $\lambda\in \F$ as the parameter. The indecomposable $d$-dimensional representations of $\F[x]$ are given exactly by the $d \times d$ Jordan blocks. } possibly with finitely many 
exceptions. While this can be much more complicated than finite type, it is still ``classifiable;'' such algebras are said to be of \emph{tame type}.\footnote{Note that finite type can be considered as a special case of tame type, namely 
when the number of one-parameter families is $0$. In the literature, some 
authors take the definition of ``tame type'' to explicitly exclude 
finite type. We do not adopt that approach here.} Finally, some algebras $L$ have 
the surprising property that any indecomposable representation of \emph{any} algebra can be 
``embedded as'' (or ``simulated by'') an indecomposable of $L$; such algebras are 
called \emph{wild}. Drozd's celebrated dichotomy theorem \cite{Drozd80} says that 
every algebra over an algebraically closed field is either tame or wild.

In the case of groups, there is an explicit description of the three cases (see 
\cite[Theorem 4.4.4]{Ben98}): let $p$ be the characteristic of the field $\F$. $\F G$ 
is of finite type \Iff $p=0$, or $p>0$ and the Sylow $p$-subgroup of $G$ is 
cyclic. $G$ is of tame type, but not finite, \Iff $p=2$ and the Sylow $2$-subgroup 
of $G$ is dihedral, semi-dihedral, or generalized quaternion (see 
\Sec{sec:prelim} for definitions). All other cases are wild.

Suppose a group $G$ has a normal subgroup $A$ isomorphic to $\Z_p^d$, and let 
$Q=G/A$. $G$ is called a \emph{tame extension} of $A$ by $Q$, if 
$\overline{\F}_pQ$ 
is of tame type.\footnote{$\overline{\F}_p$ is the algebraic closure of $\F_p$. Although it is not 
standard to apply ``tame'' to extensions, this slight abuse is justified by the mathematical results behind our main theorem. }
We solve Problem~\ref{prob:key} exactly for groups of this form. Note that the 
Sylow $p$-subgroup being cyclic already generalizes the condition for $\cite{BQ}$. 

\begin{theorem}\label{thm:main}
Suppose $G, H$ come from the class of groups that have characteristic subgroups of 
the same type and isomorphic to the
elementary Abelian subgroup $\Z_p^d$. There is a polynomial-time algorithm to compute the coset of isomorphisms $\Iso(G, H)$ from the coset of isomorphisms $\Iso(G / 
\Z_p^d, H / \Z_p^d)$, if $G$ is a tame extension of $\Z_p^d$, namely if the Sylow 
$p$-subgroups of $G/\Z_p^d$ are cyclic, dihedral, semi-dihedral, or generalized quaternion.
\end{theorem}

The condition on $G/\Z_p^d$ is satisfied by several well-known group classes:
\begin{itemize}
\item Groups with dihedral Sylow $2$-subgroups are classified \cite{gorensteinWalterAll,bender}: Let $\mathrm{O}(G)$ be the maximal normal odd-order subgroup. If $G$ has 
a dihedral Sylow subgroup, $G/\mathrm{O}(G)$ must be isomorphic to one of: (i) a subgroup of $\text{P}\Gamma\text{L}_2(\F_q)$ containing $\PSL_2(\F_q)$;\footnote{$\text{P}\Gamma\text{L}_n(\F_q)$ is the semi-direct product $\PGL_n(\F_q) \rtimes \Gal(\F_q / \F_p)$, where the Galois group $\Gal(\F_q / \F_p)$ acts on $n \times n$ matrices by sending each entry $\alpha$ to $\alpha^p$, where $p$ is the unique prime dividing $q$.} (ii) the alternating group $\mathrm{A}_7$; 
(iii) a Sylow 2-subgroup of $G$. 

\item The Sylow 2-subgroup of $\SL_2(\F_q)$ is generalized quaternion when $q$ is 
odd \cite[p.~42]{gorenstein} (or see \cite[Corollary~4.12]{conradNotes}).

\item If $D$ is a division ring, then any Sylow subgroup of a finite subgroup of the unit group $D \backslash \{0\}$
is cyclic or generalized quaternion (see \cite[Corollary~4.10]{conradNotes}). 

\item The Sylow 2-subgroups of the following groups are semi-dihedral: 
$\PSL_3(\F_q)$ for $q \equiv 3 \pmod{4}$, $\PSU_3(\F_q)$ for $q \equiv 1 
\pmod{4}$, the Mathieu group $M_{11}$, and $\GL_2(\F_q)$ for $q \equiv 3 \pmod{4}$ (see, \eg, \cite{ABG}).
\end{itemize}

Theorem~\ref{thm:main} allows us to solve \GpI in $\P$ for a class of groups that we now describe. Following \cite{BQ}, we say that a group $G$ has a \emph{Sylow 
tower} if there is 
a normal series $\id = G_\ell \lhd \dotsb \lhd G_1 \lhd G_0 = G$ where each $G_i 
/ 
G_{i+1}$ is isomorphic to a Sylow subgroup of $G$. We say that $G$ has an 
\emph{elementary Abelian Sylow tower} if furthermore all its Sylow 
subgroups are elementary Abelian. 

\begin{corollary} \label{cor:main}
The coset of isomorphisms 
between two groups $G, H$ can be computed in polynomial time when (1) $\rad(G)$ 
has an elementary 
Abelian Sylow tower, and (2) for any prime $p$ dividing $|\rad(G)|$, the Sylow 
$p$-subgroup of $G/\rad(G)$ is cyclic, dihedral, semi-dihedral, or generalized 
quaternion.
\end{corollary}

\begin{proof} 
The algorithm of \cite{BCGQ} computes the coset of isomorphisms for groups of the 
form $G / \rad(G)$. Apply Theorem~\ref{thm:main} iteratively, with this as the 
base case. 

To ensure that the condition is satisfied iteratively, we need the following fact. Let $\id = G_\ell \lhd \dotsb \lhd 
G_1 \lhd G_0 = G$ be the filtration 
where $G_1=\rad(G)$ and the rest is an elementary Abelian Sylow tower. Suppose at 
the $i$th step, $i\geq 1$, $G_i/G_{i+1}$ is an elementary 
Abelian $p_i$-group. We need to show that the Sylow $p_i$-subgroup $G/G_i$ is 
isomorphic to the Sylow $p_i$-subgroup of $G/G_1$. 

The preceding fact follows from the claim: If $N$ is a normal subgroup of $G$, and 
$p\nmid |N|$, then a Sylow $p$-subgroup of $G/N$ is isomorphic to a Sylow 
$p$-subgroup of $G$. The claim follows from the Schur--Zassenhaus Theorem, 
but there is also a more direct, elementary proof, as follows. Let $P$ be a Sylow 
$p$-subgroup of $G$, and consider the restriction of the quotient map 
$\varphi\colon G \twoheadrightarrow G/N$ to $P$. Since $p \nmid |N|$, $P \cap N = 
1$, so $P$ is mapped isomorphically onto its image in $G/N$. Since $p \nmid |N|$, 
$|P|$ is the largest power of $p$ dividing $|G/N|$, so the image of $P$ under the 
quotient map $G \twoheadrightarrow G/N$ is a Sylow $p$-subgroup of $G/N$.
\end{proof}

We now compare our result with the previous one in \cite{BQ}. 
Firstly, a 
critical difference is that in our setting we need to deal with both actions and 
cohomology classes (see \Sec{sec:strategy}). In the setting of \cite{BQ}, the
Schur--Zassenhaus theorem implies that the cohomology classes are always trivial, so this part 
does not appear in \cite{BQ} at all. 
Secondly, to deal with actions (Problem~\ref{prob:action}), though we follow the 
algorithmic framework of \cite{BQ}, for the supporting algorithmic subroutines, we 
need to use some sophisticated algorithms in computational algebra (see 
\Sec{sec:prelim}), while in \cite{BQ} the corresponding subroutines are rather 
straightforward. Finally, we bound the running time of our algorithms by proving size bounds on representations and on group cohomology in the tame case, using an explicit description of representations from the literature, and using previously known results on group cohomology.
This was not needed in \cite{BQ}.

More broadly, to achieve Theorem~\ref{thm:main}, for the first time in the worst-case analysis of \GpI, we step into the regime of modular representation theory---that is,
when the characteristic of the underlying field divides the order of the group. 
This theory is much less 
well-understood than ordinary representation theory. As the reader may see 
later, to solve Problem~\ref{prob:key} in general seems to require certain 
deep use of this theory.  
We hope this article serves as a first step in this direction.

\paragraph{Organization.} We first 
present some preliminaries in 
\Sec{sec:prelim}. In \Sec{sec:strategy} we show how the splitting strategy of 
\cite{GQccc} applies in this case, and in \Sec{sec:proofOverview} we give an 
overview of the proofs. Detailed proofs for the action aspect and the cohomology 
aspect are presented in \Sec{sec:action} and \Sec{sec:cohomology}. 
Finally, in \Sec{sec:conclusion} we discuss the general relationship between 
\GpI and the tame-wild dichotomy in representation theory, and present some open 
questions. The appendix is devoted to reproduce Crawley-Boevey's description of 
the indecomposable modules of semi-dihedral algebras for readers' convenience. 

\section{Preliminaries} \label{sec:prelim}

\paragraph{Notations and definitions.} For a prime $p$, $\F_p$ denotes the 
field of size $p$. 
The characteristic of a field $\F$ is 
denoted $\Char(\F)$. $\M(n, p)$ is the set of $n\times n$ matrices over 
$\F_p$, and $\GL(n, p)$ is the group of $n\times n$ invertible matrices of $\F_p$. 
For $n\in\N$, 
$[n]:=\{1, \dots, n\}$. $\sym(\Omega)$ denotes the symmetric group over a set $\Omega$; when $\Omega=[n]$ we write $\sg_n$. A permutation group over $\Omega$ is a 
subgroup of $\sym(\Omega)$. 

$\Z_p$ denotes the cyclic group of order $p$. A 
group is elementary Abelian if it is isomorphic to $\Z_p^d$ for some prime $p$ and 
some integer $d$. 
The dihedral groups (of order a power of $2$) are $\dihedral{m}=\langle x, y \mid x^2=y^{2^m}=1, 
yx=xy^{-1}\rangle$.
The semi-dihedral or quasi-dihedral groups are $\sdihedral{m}=\langle x, y\mid x^2=y^{2^m}=1, yx=xy^{2^{m-1}-1}\rangle$.
The (generalized) quaternion groups are $\gquaternion{m}=\langle x, y\mid 
x^2=y^{2^{m-1}}, yx=xy^{-1}\rangle$. 
$\dihedral{m}$, $\sdihedral{m}$, and $\gquaternion{m}$ are of order $2^{m+1}$; $\dihedral{1}$ is the Klein four group. 

\begin{remark}
There is a polynomial-time algorithm to decide whether a given group is $\dihedral{m}$, $\sdihedral{m}$, or $\gquaternion{m}$, because these groups are 
generated by two elements.
\end{remark}

\paragraph{General group theory.} 
A $p$-group for $p$ prime is a group whose 
order is $p^d$ for some $d$. A Sylow $p$-subgroup of a group $G$ is a maximal 
$p$-subgroup of $G$, under inclusion. Two of the Sylow theorems say that every 
finite group has a Sylow $p$-subgroup whose order is the largest power of $p$ 
that 
divides $G$, and all Sylow $p$-subgroups of $G$ are conjugate to one another. 
Thus, up to isomorphism, we may speak of ``the'' Sylow $p$-subgroup of a group 
$G$. 
Given the Cayley table of a group, a Sylow $p$-subgroup can be found in 
polynomial time.

A subgroup $N$ of $G$ is 
\emph{characteristic} if $N$ is sent to itself by every automorphism of $G$. A 
\emph{characteristic subgroup functor} is a function $\mathcal{S}$ from finite 
groups to finite groups such that (1) $\mathcal{S}(G) \leq G$ for all $G$, and (2) 
any isomorphism $\varphi \colon G_1 \to G_2$ restricts to an isomorphism 
$\varphi|_{\mathcal{S}(G_1)} \colon \mathcal{S}(G_1) \to \mathcal{S}(G_2)$. In 
particular, it follows that $\mathcal{S}(G)$ is always characteristic in $G$. 
Examples of characteristic subgroup functors include most ``natural'' 
characteristic subgroups such as the center, the derived subgroup, and the terms 
of the derived, lower central, and upper central series. A characteristic 
subgroup 
functor is Abelian (resp. elementary Abelian), if $\mathcal{S}(G)$ is Abelian 
(resp., elementary Abelian) for all $G$. 

\noindent\textit{Convention:} In this paper, whenever we say ``characteristic 
subgroup'' we 
mean the image of an implied characteristic subgroup functor.

\paragraph{Indecomposable modules.} As representations of a group $Q$ 
over a field $\F$ are the same as modules over the group algebra $\F Q$, 
we shall use 
the terms module and representation interchangeably. For two representations 
$\theta$ and $\eta$, we use $\theta\cong \eta$ to denote that they are 
equivalent. Let $M$ be a module of an 
algebra $L$. $M$ is 
\emph{indecomposable} if it cannot be written as a direct sum of two submodules. We denote the set of $d$-dimensional indecomposable modules of an algebra $L$ by $\Ind(L,d)$.
The decomposition of $M$ into a direct sum of indecomposables is essentially unique:

\begin{theorem}[{Krull--Schmidt (see, \eg, \cite[Theorem 1.4.6]{Ben98})}]\label{thm:krull_schmidt}
Let $\phi$ and $\psi$ be two linear representations of a group $Q$. Suppose
$\phi=\iota_1^{d_1}\oplus \dots \oplus \iota_\ell^{d_\ell}$ and 
$\psi=\iota_1^{e_1}\oplus \dots \oplus \iota_\ell^{e_\ell}$, where $\iota_i$'s 
are indecomposable and pairwise non-isomorphic, and all $d_i,e_i \geq 0$. Then $\phi \cong \psi$ \Iff $d_i=e_i$ for every $i\in[\ell]$.
\end{theorem}

\paragraph{2-cohomology classes.} 
Let $Q$ be 
a group, and $A$ an Abelian group. An action $\theta$ of $Q$ on $A$ is a group 
homomorphism $Q\to\aut(A)$. A \emph{2-cocycle} with respect to the action 
$\theta$ is a function $f:Q\times Q\to A$ satisfying the 2-cocycle identity $f(p,q) + f(pq,r) = \theta_p(f(q,r)) + f(p,qr)$. The 
set of all 2-cocycles is an Abelian group under pointwise addition, denoted $Z^2(Q, A, \theta)$. Given a function $u:Q\to A$, the function
$b_u(q, q')=u(q)+\theta_q(u(q'))-u(qq')$ is a \emph{2-coboundary} $b_u:Q\times Q\to A$. The set 
of 2-coboundaries is a subgroup of $Z^2(Q, A, \theta)$, denoted $B^2(Q, A, 
\theta)$. The quotient group $H^2(Q, A, \theta):=Z^2(Q, A, \theta)/B^2(Q, A, 
\theta)$ is the group of \emph{2-cohomology classes}. For $f$ and $g$ in $Z^2(Q, 
A, \theta)$, if $f-g\in B^2(Q, A, \theta)$ (representing the same cohomology 
class), they are called \emph{cohomologous}, denoted $f \simeq g$.

\paragraph{Preliminaries for algorithms.} As customary in 
permutation 
group algorithms \cite{seressbook}, a permutation group is represented in 
algorithms by a set of generators. The automorphism group of a group $G$ is 
represented as a permutation subgroup of $\sym(G)$. A coset of a permutation group 
is represented 
by a single coset representative together with a set of generators for the 
subgroup.
A representation of $Q$ is given by listing the images of $q\in Q$ 
explicitly. Two representations $\theta$ and $\eta$ are \emph{equal}, denoted $\theta=\eta$, if 
$\theta(q)=\eta(q)$ for every $q\in Q$; compare with $\theta\cong\eta$.
A 2-cohomology class is represented by a 2-cocycle $f$, which 
in turn can be viewed as a matrix over $\Z_p$ of size $d\times |Q|^2$ when $A \cong \Z_p^d$. 
In the algorithm, we need to test whether two 2-cocycles $f_1$ and $f_2$ are 
cohomologous. 
This can be done as in \cite{GQccc}; for completeness we 
present a proof here. 

\begin{proposition}[{\cite{GQccc}}]\label{prop:cohomologous}
Given two 2-cocycles $f$ and $g$ \wrt the action $\theta:Q\to \aut(A)$  
($A=\Z_p^d$), whether $f \simeq g$ can be decided in time $\poly(|Q|, d, \log p)$.
\end{proposition}

\begin{proof} 
We need to check whether $f - g \in B^2(Q,A,\theta)$. For this, compute 
a basis of $B^2(Q,A,\theta)$ as a $\Z_p$-vector space: This can be done 
by applying the defining equation of 2-coboundaries to a basis of $\{u:Q\to A\}$, 
the dimension of which is $d\cdot |Q|$. As we can treat these as vector spaces, we 
then test whether $f-g$ is in the $\Z_p$-span of 
2-coboundaries (as a vector in a space of dimension $d\cdot |Q|^2$). As a 
standard algorithmic task in linear algebra, this can be solved efficiently.
\end{proof}

\begin{theorem}[{Module isomorphism 
\cite{CIK97,BL08,IKS10}}]\label{thm:module_iso}
Given two tuples of matrices $(A_1, \dots, A_n)$, $(B_1, \dots, B_n)$, $A_i, 
B_j\in M(d, p)$, there exists a deterministic $\poly(d, n, \log 
p)$-time algorithm that finds $C\in\GL(d, p)$ such that for every 
$i\in[n]$, $CA_i=B_iC$, if such $C$ exists.
\end{theorem}

A \emph{matrix algebra} is a linear subspace $L$ of $n \times n$ matrices over a field such that $L$ is closed under matrix multiplication ($a,a' \in L \Rightarrow aa' \in L$). The \emph{unit group} of a ring or algebra $A$ is the set of invertible elements in $A$, which naturally form a group under multiplication.

\begin{theorem}[{Finding units in a matrix algebra \cite{BO08}}]\label{thm:unit}
Given a linear basis of a matrix algebra $L$ in $\M(d, p)$, a generating set of the unit group of $L$ can be 
computed deterministically in time $\poly(d, p)$.
\end{theorem}

\begin{theorem}[Decomposing into indecomposables 
\cite{CIK97}]\label{thm:decompose}
Given a module $M$ over an algebra $L$ over a finite field $\F$, a direct sum 
decomposition of $M$ can be computed in time polynomial in the input size and 
$\Char(\F)$.
\end{theorem}

\begin{theorem}[Parametrized setwise transporter problem 
\cite{BQ}]\label{thm:setwise}
Given a set of generators of $P\leq \sg_t$, and $S, T\subseteq [t]$ with 
$|S|=|T|=k$, $P_{S\to T}:=\{\sigma \in P \mid S^\sigma = T\}$ can be computed in 
time $\poly(t, 2^k)$. 
\end{theorem}

\section{The divide and conquer strategy for Problem~\ref{prob:key}}\label{sec:strategy}

Now we briefly recall the divide and conquer strategy from \cite{GQccc}, and how 
it applies to the particular case of Problem~\ref{prob:key}. 
Problem~\ref{prob:key} requires us to compute isomorphisms of $G,H$ from 
isomorphisms of $G/\Z_p^d, H/\Z_p^d$. It is then natural to examine how the 
quotient group $G/\Z_p^d$ and the characteristic subgroup $\Z_p^d$ are related by 
$G$; this is the starting point for the strategy from \cite{GQccc}.

Given a group $G$ and an Abelian characteristic subgroup $A$ of 
$G$, let $Q:=G/A$; we denote this situation $A \hookrightarrow G \twoheadrightarrow Q$ and call $G$ an \emph{extension} of $A$ by $Q$. 
The \emph{extension data} of $A \hookrightarrow G \twoheadrightarrow Q$
consists of two functions: the (conjugation) \emph{action} 
$\theta:Q\times A\to A$ defined by $(q, a)\to qaq^{-1}$, and the \emph{2-cocycle} 
$f_s:Q\times Q\to A$, depending on a transversal or \emph{section}
$s:Q\to G$---\ie, an assignment of an element $s(q)$ to each coset $q \in G/A$---and defined 
by $f_s(p, q) := s(p)s(q)s(pq)^{-1}$. Note that $\aut(A)\rtimes\aut(Q)$ acts naturally on the set of actions (including $\theta$) and the set of 2-cocycles (including $f_s$). 

In Problem~\ref{prob:key}, we are given two groups $G$ and $H$, and their respective characteristic subgroups $A$ 
and $B$ (recall our convention about characteristic subgroup \emph{functors} from \Sec{sec:prelim}). Note that if $G\cong H$, then $A\cong B$ and $G/A\cong H/B$. We first 
test whether $A\cong B$; this is easy because they are Abelian. 
Recall that we are given $\Iso(G/A, H/B)$; if it is empty then $G\not\cong H$. Therefore, at this point we have either determined that $G \not\cong H$, or we have $A\cong B$ (identified as $A$), and $G/A\cong H/B$ (identified as $Q$). This is the divide step of the strategy. 

But these conditions are not sufficient to conclude $G \cong H$, so we have yet to conquer, as in the following:

\begin{example}
We give an example of two tame extensions $A \hookrightarrow G \twoheadrightarrow G/A$ and $B \hookrightarrow H \twoheadrightarrow H/B$ with $A$ characteristic in $G$, $B$ characteristic in $H$, $A \cong B$, and $G/A \cong H/B$, but $G \not \cong H$. Let $G = D_{4k} = \langle \rho, \tau | \rho^{2k} = \tau^2 = 1, \tau \rho \tau = \rho^{-1} \rangle$ be the dihedral group of order $4k$ with $k$ odd, and let $H = \Z_2 \times D_{2k}$. In both groups, the center---a characteristic subgroup---is $\Z_2$ (which also happens to be the unique maximal normal 2-group); in the case of $H$ this is clear, in the case of $G$ it is the subgroup $\{1, \rho^k \}$. Both groups are thus characteristic extensions of $\Z_2$ by $D_{2k}$. Note that the Sylow 2-subgroup of $D_{2k}$ is cyclic of order 2, since $k$ is odd, so these are both tame extensions. Yet $G \not\cong H$; this can be seen by noting that $H$ contains elements $h,x$ with $h$ of order $2k$, $x$ of order 2 such that $hx$ has order $k$, yet this is not true of $G$: The only elements in $G$ of order $2k$ are the generators of $\langle \rho \rangle$, and multiplying any of those by an element of order 2 yields another element of order 2.
\end{example}

Since every element of $G$ has a unique expression as $a s(q)$ for $a \in A, q \in Q$, $\Iso(G, H)$ embeds as a subgroup of $\aut(A)\rtimes \aut(Q)$. When $A\cong\Z_p^d$, we have $\aut(A)\cong \GL(d, p)$; $\aut(Q)$ is given to us as part of $\Iso(G/A, H/B)$. By 
\cite[Lemma~II.2]{GQccc}, $\Iso(G, H)$ consists exactly of those $(\alpha, \beta)\in\aut(A)\rtimes\aut(Q)$ that 
make the two extension data the same.\footnote{Note here that the condition of 
characteristic groups is crucial. That is, if $A$ and $B$ are merely normal 
subgroups, then this does not hold in general. See \cite{GQccc} for details. } 
Following 
\cite{GQccc}, we refer to the problem 
of computing the coset in $\aut(A)\times \aut(Q)$ consisting of elements sending one extension to 
the other as \edpcl (or \edpcs):

\begin{problem}\label{prob:edpc}
Let $A\cong \Z_p^d$. Given $\aut(Q)$ and the extension data $(\theta, f)$ and $(\eta, g)$ of  
$A\hookrightarrow G\twoheadrightarrow Q$ and $A\hookrightarrow 
H\twoheadrightarrow Q$, respectively, compute $\{(\alpha, 
\beta)\in\aut(A)\times \aut(Q) : \theta^{(\alpha, \beta)}=\eta, \text{ and } 
f^{(\alpha, \beta)}\simeq g\}$.
\end{problem}

On first sight, \edpcs asks for $(\alpha, \beta)$ that sends $\theta$ to $\eta$ and 
$f$ to $g$, \emph{simultaneously}. However, note that $f\in H^2(Q, A, \theta)$; that is, 
to define the space in which $f$ lives relies on $\theta$ in the first place. On the other hand, 
$\theta$ has no dependence on $f$. Therefore, \edpcs reduces to solving the 
following two problems, \emph{in order}:

\begin{problem}\label{prob:action}
Suppose we are given a group $Q$ by its Cayley table, 
$\aut(Q)$ by a set of generators, and two linear representations $\theta, 
\eta:Q\to 
\GL(d, p)$ by listing images of $Q$ explicitly. Compute a set of generators for 
the coset $\{(\alpha, \beta)\in\GL(d, p)\times \aut(Q) \mid \theta^{(\alpha, 
\beta)}= \eta\}$, in time $\poly(|Q|, p^d)$.
\end{problem}

\begin{problem}\label{prob:coh_iso}
Suppose we are given a group $Q$ by its Cayley table, a representation 
$\theta\colon Q \to \GL(d,p)$ by listing the images of $Q$ explicitly, and two 
2-cocycles $f,g\colon Q \times Q \to \Z_p^d$ in $Z^2(Q, \Z_p^d, \theta)$. 
Furthermore we are given a set of generators for $\{(\alpha, \beta)\in\GL(d, 
p)\times \aut(Q)\mid 
\theta^{(\alpha, \beta)}=\theta\}$. Compute a 
set of generators for the coset $\{(\alpha, \beta) \in \GL(d, p)\times \aut(Q) 
\mid  
f^{(\alpha,\beta)} \simeq g \text{ and } \theta^{(\alpha,\beta)} = \theta\}$, in time $\poly(|Q|,p^d)$.
\end{problem}

We shall refer to Problem~\ref{prob:action} as \actcompl (or \actcomps), 
and Problem~\ref{prob:coh_iso} as \cohisol (or \cohisos).

\section{Overview of algorithms for \actcomp and \cohiso} \label{sec:proofOverview}
In this section we give an overview of the algorithms for \actcomps and \cohisos 
when $\overline{\F}_pQ$ is tame, thereby proving Theorem~\ref{thm:main}. The complete proof for \actcomps is in Section~\ref{sec:action} and for \cohisos is in Section~\ref{sec:cohomology}.

The algorithm for \actcomps goes as follows: given representations $\theta, 
\eta:Q\to\GL(d, p)$, first decompose them into a direct sum of indecomposables 
(Theorem~\ref{thm:decompose}), 
and group them by isomorphism types (Theorem~\ref{thm:module_iso}). That is, 
$\theta=\iota_1^{d_1}\oplus 
\iota_2^{d_2}\oplus \dots\oplus \iota_\ell^{d_\ell}$, and 
$\eta=\iota_1^{e_1}\oplus \iota_2^{e_2}\oplus\dots\oplus\iota_\ell^{e_\ell}$. 
(Some $d_i$'s and/or $e_j$'s may be $0$.) By Theorem~\ref{thm:krull_schmidt}, 
$\theta\cong \eta$ \Iff  
$d_i=e_i$ for all $i\in[\ell]$. To take into account the 
effect of $\aut(Q)$, consider the induced action of $\aut(Q)$ on the
indecomposables of $\F_p Q$. Firstly, 
compute the closure of $\Iota=\{\iota_1, \dots, \iota_\ell\}$ under $\aut(Q)$---that is, the set of all indecomposables that are in the $\aut(Q)$-orbit of any $\iota_i$---denoted $\clo(\Iota)$. Viewing $\aut(Q)$ as a permutation group on the domain 
$\clo(\Iota)$, we need to compute the coset in $\aut(Q)$ that sends those 
indecomposables in $\theta$ of multiplicity $m$, to those indecomposables in 
$\eta$ of multiplicity $m$, for every $m\in[d]$. For each $m\in[d]$, this is a 
setwise 
transporter problem, so applying Theorem~\ref{thm:setwise} sequentially gives 
an efficient 
algorithm---provided that we can upper bound the number of indecomposables of 
dimension $d$, and thereby $|\clo(\Iota)|$, by $\poly(|Q|, p^d)$. We prove that for the tame type this holds 
(\Sec{subsec:bound_act}), and for wild type it always fails (\Sec{subsec:wild}). This does not follow 
directly from the definition of the tame--wild dichotomy, since that requires the underlying field 
to be infinite, whereas we care about 
representations over a \emph{finite} field and also need an upper bound on the 
\emph{number} of indecomposables. We are nonetheless able to prove the upper bound we need by using the explicit description of the indecomposable families for tame group algebras due to Crawley-Boevey \cite{CB89}.
This may be viewed as 
the first main technical contribution of this work. On the other hand, by 
\cite{Rickard}, for the wild type this upper bound fails badly (see \Sec{subsec:wild}). Finally, by 
Theorem~\ref{thm:unit} and~\ref{thm:module_iso} we can compute, 
for each $\beta\in\aut(Q)$ that make $\theta$ and $\eta$ isomorphic, the coset 
$\alpha\in\GL(d, p)$ that make $\theta^{(\alpha, \beta)}=\eta$.

We then give an algorithm for \cohisos that takes the coset of action 
compatibilities as its input. As for \actcomps, the idea is to view the group of 
action compatibilities as a permutation group on $H^2(Q, \Z_p^d, \theta)$. Then 
given two 2-cocycles (representing two 2-cohomology classes), the problem becomes 
a pointwise transporter problem, a classical problem in permutation group 
algorithms that is polynomial-time solvable \cite{seressbook}. For this algorithm 
to be efficient 
in our setting, we need to upper bound 
$|H^2(Q, \Z_p^d, \theta)|$ as $\poly(|Q|, p^d)$ when $\overline{\F}_pQ$ is tame. 
Using some 
standard cohomological yoga combined with known but 
deep results on group cohomology \cite{GKKL}, we show that, amazingly, this is 
true. This is the second main technical contribution of this work. This finishes the overview.

\section{Algorithm and bounds for \actcompl}\label{sec:action}

In this section we give full details for solving the \actcomp in the tame case. 

\subsection{Algorithm for the coset of action compatibilities}
To test equivalence of two linear representations over $\F_p$, by 
Theorem~\ref{thm:krull_schmidt} we just need to 
compare the multiplicities of the corresponding indecomposables. The difficulty 
now is how to take into account the effects of $\aut(Q)$. To tackle this, the key 
idea is to view $\aut(Q)$ as a permutation group on 
a domain consisting of indecomposable representations. Let $\Ind(Q)$ be the set of 
indecomposable representations of $G$ up to equivalence. For $S\subseteq \Ind(Q)$, 
we use $\clo(S)$ to denote the closure of $S$ under $\aut(Q)$. By applying 
generators of $\aut(Q)$ iteratively and checking whether new indecomposables are 
generated or not using Theorem~\ref{thm:module_iso}, we have the following breadth-first-search-style algorithm:
\begin{proposition}\label{prop:closure}
Given $S\subseteq \Ind(Q)$, $\clo(S)$ can be computed in time $\poly(|\clo(S)|)$.
\end{proposition}

A trivial upper bound for $|\clo(S)|$ is $|\Ind(Q)|$, the total number of 
indecomposables of $\F_pQ$. Another natural bound for $|\clo(S)|$ utilizes the 
dimensions of indecomposables in $S$. Suppose $S$ is finite and $\{d_1, \dots, 
d_\ell\}$ are the dimensions of indecomposables in $S$. Then $|\clo(S)|$ is upper 
bounded by the sum of the number of indecomposables of dimensions $d_1, \dots, 
d_\ell$, denoted $\Ind(Q, d_1), \dotsc, \Ind(Q, d_\ell)$.

\begin{theorem} \label{thm:action}
Problem~\ref{prob:action} can be solved for representations of $Q$ over $\F_p$ of 
dimension $d$ when the number of indecomposable $\F_p Q$-modules of dimension $d$ 
is bounded by $\poly(|Q|,p^d)$ for all $d$.
\end{theorem}

For the proof, we need one more straightforward observation:

\begin{observation}\footnote{Essentially this observation appeared as 
\cite[Claim~1]{BQ}, but that formulation as only for direct products, not 
semi-direct products, and there was a typo in its formulation there. This 
observation is also used in the journal version of \cite{GQccc}. We include the 
short proof here for completeness.} \label{obs:genSubProd}
Let $G$ be a subgroup of $H \rtimes K$, let $\pi_{K} \colon G \to K$ denote the 
natural projection onto $K$ with kernel $H$, and let $G_H$ denote the intersection 
$G \cap (H \rtimes 1)$. If $\mathcal{H} \subseteq H \rtimes 1$ generates $G_H$ and 
$\mathcal{K} \subseteq K$ generates $\pi_K(G)$, and for each $k \in \mathcal{K}$, 
$h_k$ is such that $(h_k, k) \in G$, then $\mathcal{H} \cup \{(h_k, k) \in G : k 
\in \mathcal{K}\}$ generates $G$.
\end{observation}

\begin{proof}
Given $(h,k) \in G$, first we write $k$ as a word in the generators $\mathcal{K}$, 
say $k = k_1 \dotsb k_\ell$, with each $k_i \in \mathcal{K}$. Then $(h,k)\cdot 
\left( (h_{k_1}, k_1) (h_{k_2}, k_2) \dotsb (h_{k_{\ell}}, k_\ell) \right)^{-1}$ 
is of the form $(h', 1)$, which is in $G_H$. Write $(h', 1)$ as a word in 
$\mathcal{H}$. 
\end{proof}

\begin{proof}[Proof of Theorem~\ref{thm:action}]
Given two $d$-dimensional representations $\theta, \eta$ of $Q$ over $\F_p$, use 
Theorem~\ref{thm:decompose} and Theorem~\ref{thm:module_iso} to decompose and group by 
isomorphism types as $\phi=\iota_1^{d_1}\oplus \dots \oplus \iota_\ell^{d_\ell}$, 
and 
$\psi=\kappa_1^{e_1}\oplus\dots\oplus\kappa_{\ell'}^{e_{\ell'}}$. Let  
$\Ind(\phi)=\{\iota_1, \dots, \iota_\ell\}$, and similarly we have $\Ind(\psi)$. 

For any $\beta\in\aut(Q)$, $\theta^{\beta}=(\iota_1^\beta)^{d_1}\oplus \dots 
\oplus (\iota_\ell^\beta)^{d_\ell}$. If $\theta^{\beta}\cong \eta$, then by  
Theorem~\ref{thm:krull_schmidt}, $\ell=\ell'$ and there 
exists $\sigma\in\sg_{\ell}$ \st $\iota_{\sigma(i)}^\alpha=\kappa_i$ and 
$d_{\sigma(i)}=e_i$. Furthermore, this also implies that 
$\clo(\Ind(\phi))=\clo(\Ind(\psi))$. 

Given $\theta=\iota_1^{d_1}\oplus \dots \oplus \iota_\ell^{d_\ell}$, and 
$\eta=\kappa_1^{e_1}\oplus\dots\oplus\kappa_{\ell'}^{e_{\ell'}}$, first check 
whether $\ell=\ell'$ and $\clo(\Ind(\theta))=\clo(\Ind(\eta))$. If either of these 
two conditions is not satisfied, then $\theta$ and $\eta$ cannot be equivalent 
under any 
$\beta\in\aut(Q)$. If these two conditions are satisfied, let 
$\Omega=\clo(\Ind(\theta))$. The action of $\aut(Q)$ on $\Omega$ allows us to 
consider $\aut(Q)$ (or, more precisely, its homomorphic image)
as a permutation group $A_Q \leq \sym(\Omega)$. View $\phi$ and $\psi$ as 
functions from $\Omega$ to 
$\N$, that is, $\phi(\iota)$ is the multiplicity of $\iota$ in $\phi$. Our task 
now is just to decide whether there exists $\sigma\in A_Q$ \st for each 
multiplicity $m$, $\sigma$ sends those indecomposables in $\phi$ of multiplicity 
$m$ to those indecomposables in $\psi$ of multiplicity $m$. This is 
clearly a set-wise transporter problem. Solve this iteratively for each 
multiplicity. This gives a generating set for the coset $\{\beta\in\aut(Q)\mid 
\theta^\beta\cong \eta\}$. For each $\beta$ in the generating set, use 
Theorem~\ref{thm:module_iso} to compute $C\in\GL(d, p)$ \st $C\theta 
C^{-1}=\eta$, and use Theorem~\ref{thm:unit} to compute the unit group of $\eta$. 
Collect all the generators, which gives a generating set for $\{\GL(d, p)\times 
\aut(Q)\mid \theta^{(\alpha, \beta)}=\eta\}$. This finishes the description of the 
algorithm. 

Decomposing the representations takes time polynomial in $d$ and $p$, by 
Theorem~\ref{thm:decompose}, which is much better than what we need for our purposes. The 
application of the setwise transporter algorithm (Theorem~\ref{thm:setwise}) takes 
time $\poly(|\Omega|, 2^d) 
= \poly(|\clo(\Ind(\phi))|, p^d)$. If the number of indecomposable modules of 
dimension $d'$ is bounded by $\poly(|Q|,p^{d'})$ for all $d'$, and $\dim \iota_i = 
d_i$, then $|\clo(\Ind(\phi))|$ is bounded by $\poly(|Q|,\sum_i p^{d_i}) \leq 
\poly(|Q|,p^d)$. Thus this application of the setwise transporter algorithm takes 
time polynomial in the input size.

This gives us one element of $\aut(Q)$ that sends $\theta$ to $\eta$ (up to 
equivalence), as well as generators of the subgroup of $\aut(Q)$ that sends 
$\theta$ to itself (up to equivalence). To get the actual coset of action 
compatibilities, we need a subgroup of $\aut(A) \rtimes \aut(Q)$, that is, 
including data about the linear equivalences. By Observation~\ref{obs:genSubProd}, 
it is enough to find, for each generator $\alpha$ of the subgroup of $\aut(Q)$, 
generators of the subgroup $\{ \beta \in \aut(A) : \theta^{(\alpha, \beta)} = 
\theta\}$ (note, \emph{equality} here, not merely equivalence). 

To do this, we first find a linear spanning set of the linear subspace of $M(d,p)$ 
consisting of those matrices $\beta$ such that $\beta \theta(q) = 
\theta(q^{\alpha}) \beta$, using linear algebra over $\Z_p$. This linear subspace 
is in fact closed under matrix multiplication, as one can easily check, and the 
subgroup of $\aut(A)$ we seek is just the group of units of this matrix algebra. 
From the matrix algebra itself, we can find its group of units in polynomial time 
(Theorem~\ref{thm:unit}).
\end{proof}

For future reference we highlight the key criterion needed for the preceding 
algorithm to run efficiently:
\begin{criterion}\label{criterion:key}
For a group $Q$, there are at most $\poly(|Q|,p^{d})$ $d$-dimensional 
indecomposable $\F_p Q$-modules.
\end{criterion}

\subsection{The number of indecomposable modules of group 
algebras}\label{subsec:bound_act}
We will show that Criterion~\ref{criterion:key} holds in the tame case, and fails quite badly for \emph{all} wild group extensions.

It should be 
noted that our results do not follow 
directly from the tame--wild dichotomy, because we need explicit upper bounds over 
finite 
fields, whereas the dichotomy is typically stated over algebraically closed 
fields and does not provide quantitative bounds. Therefore, we are forced to use the 
explicit descriptions of tame group algebras to get such quantitative bounds over finite fields.

Bounds for the case of finite representation type are furnished by the 
following theorem:

\begin{theorem}[{Higman \cite{Hig54}; see \cite[Theorem~64.1]{CR66}}]
For a group $Q$, the group algebra $\F_p Q$ is of finite representation type if 
and only if the 
Sylow $p$-subgroup of $Q$ is cyclic. If this holds, then the number of 
indecomposable $\F_p Q$-modules is $\leq |Q|$.
\end{theorem}

For tame representation type, we require a more detailed analysis. To start with, 
it is 
well-known that the representation type of $\F_p Q$ depends on Sylow $p$-subgroups 
of $Q$, even with quantitative bounds:

\begin{proposition}[{See \cite[Proposition 3(1)]{BD82}}]\label{prop:sylow}
Let $P$ be a Sylow $p$-subgroup of $Q$. Then 
$|\Ind(\F_p Q, d)|\leq [Q:P]\cdot(\sum_{d'=\lceil d/[Q:P]\rceil}^d|\Ind(\F_p P, 
d')|)$.
\end{proposition}
\begin{proof}
Any indecomposable $d$-dimensional $\F_p Q$-module $M$ is an $\F_p Q$-direct 
summand 
of some $N^Q$, where $N$ is an indecomposable $\F_p P$-module of dimension $d'$, 
with
$\frac{d}{[Q:P]}\leq d'\leq d$. (Recall that $N^Q$ denotes the induced module of 
$N$ to $Q$.) If $N$ is of dimension $d'\leq d$, then $N^Q$ contributes at most 
$[Q:P]$ non-isomorphic indecomposable $\F_p Q$-modules of dimension $d$. The claim 
then follows.
\end{proof}

In other words, to show that $Q$ satisfies Criterion~\ref{criterion:key}, it 
suffices to show 
that its Sylow $p$-subgroup $P$ satisfies Criterion~\ref{criterion:key}. 

Now we need to provide an explicit upper bound for the tame group algebras. We 
do this for the semi-dihedral groups $\sdihedral{m}$. The dihedral groups 
$\dihedral{m}$ can be deduced similarly because the structure of its 
indecomposables 
are very similar to those of $\sdihedral{m}$. (In fact, the forms of 
indecomposables for $\dihedral{m}$ are a subset of the forms for $\sdihedral{m}$. 
See \cite[Chap. 4.11]{Ben98} and compare with Appendix~\ref{app:SD}.) The 
generalized quaternion groups 
$\gquaternion{m}$ are handled by the following proposition, as 
$\gquaternion{m}$ is a subgroup of index $2$ of $\sdihedral{m+1}$: it is the 
subgroup generated by $x^2$ and $xy$. Note that this constant 
$2$ is important here.

\begin{proposition}[{See \cite[Proposition 3(2)]{BD82}}]\label{prop:subgroup}
Let $H$ be a subgroup of $Q$. Then 
$$|\Ind(\F_p H, d)|\leq [Q:H]\cdot (\sum_{d'=d}^{d\cdot[Q:H]}|\Ind(\F_p Q, d')|).$$
\end{proposition}

\begin{proof}
Any indecomposable $M$ of $\F_p H$ of dimension $d$ is a direct summand of the 
restriction of some 
$\F_p Q$ indecomposable $N$ of dimension $\leq [Q:H]\cdot d$. 
Each such $N$ 
contributes at most $[Q:H]$ non-isomorphic $\F_p H$ indecomposables of dimension 
$d$. The result then follows. 
\end{proof}

\begin{proposition}\label{prop:satisfy_key}
$\F_2 \sdihedral{m}$ satisfies the key criterion. 
\end{proposition}
\begin{proof}
We shall follow the description of Crawley-Boevey \cite{CB89}. For the reader's 
convenience his result is reproduced in Appendix~\ref{app:SD}. Though we try to be 
self-contained here, a cautious 
reader is suggested to at least go over Appendix~\ref{app:SD} briefly and return to 
this proof, since the proof ultimately builds on counting explicitly the 
specific forms from Crawley-Boevey's construction. 

To start with, since Crawley-Boevey's description 
works over fields of size $>2$, we shall consider $\F_4\sdihedral{m}$ instead of  
$\F_2\sdihedral{m}$. Indeed, as any representation over $\F_2$ is one over $\F_4$ 
via the field extension, any upper bound on the number of representations over 
$\F_4$ will be an upper bound for the number of representations over $\F_2$.

The indecomposables of the group algebras $\F_4\sdihedral{m}$ are most easily 
described in terms of the indecomposables of the so-called semi-dihedral 
algebra $\Lambda_\ell=\F_4\langle a, b\mid a^3=b^2=0, a^2=(ba)^\ell b\rangle$, 
where $\ell=2^{m-1}-1$. This is because all indecomposables except the regular one 
of 
$\F_4\sdihedral{m}$ are in one-to-one correspondence with those of $\Lambda_\ell$ 
\cite{BD82}. 

Briefly speaking, there are four classes of indecomposable modules of 
$\Lambda_\ell$, called asymmetric strings, symmetric strings, asymmetric 
bands, and symmetric bands. Each class is associated with a family of 
configurations, and an auxiliary algebra of finite type. There is a procedure 
that takes one configuration and one indecomposable module of the auxiliary 
algebra and produces an indecomposable $\Lambda_\ell$-module. Crawley-Boevey 
proved that each indecomposable $\Lambda_\ell$-module 
can be generated by this procedure, and two indecomposables with different 
configurations or different auxiliary indecomposables are non-isomorphic. 
Therefore it is enough to deduce an upper bound on the number of indecomposable 
$\Lambda_\ell$-modules from Crawley-Boevey's description. 

Let us detail the case of symmetric strings. 
To describe the configurations of the symmetric strings, consider words in the 
alphabet $\{a_i, b_j\mid i\in\{-(\ell+1), \dots, \ell+1\}, j\in\{-1, +1\}\}$, 
satisfying 
the following conditions: (1) the letters alternate between $a_i$'s and $b_j$'s; 
(2) there are no 
subwords of the form $b_1a_mb_1$, $a_{m+1}b_1$, $b_1a_{m+1}$, or $a_ib_1a_j$ where 
$i, 
j>0$. For a letter $c_i$ ($c=a$ or $b$), $c_i^{-1}=c_{-i}$. For a word $w=w_1\dots 
w_n$, define $w^{-1}=w_n^{-1}\dots w_1^{-1}$. Now impose an equivalence relation 
by identifying $w$ with $w^{-1}$. If $w=w^{-1}$ then call $w$ symmetric; otherwise 
$w$ is asymmetric. 

The configurations of symmetric strings are derived from the symmetric words. The 
auxiliary algebra associated with of symmetric strings is
$\F_4[e]/(e^2=e)$, which has only two indecomposables $E$, both are of dimension 
$1$ with 
$e$ acting as identity and $0$, respectively. Therefore, the auxiliary algebra 
associated with symmetric 
strings does not play a major role. In contrast, for bands the associated 
algebra will contribute a notable factor. 

Let us consider the case of $e$ being the identity. Given a symmetric word 
$w=za_0z^{-1}$, the rule to construct a 
$\Lambda_\ell$-module $M$ is explained in Appendix~\ref{app:SD}. Let $d=\dim(M)$. 
From there it is seen 
that 
$M$ is determined by a quiver (a directed graph) with $d$ vertices. The arrows 
(edges) are determined by $z$. To get an upper bound on the number of 
indecomposables from symmetric strings of 
dimension $d$, it is enough to note that the vertices can be arranged in a line, 
with some special gadgets. To start with, note that the arrows among the rightmost 
$2\ell+3$ arrows are fixed due to the $e$-gadget. Then, depending on whether the 
leftmost arrow is labeled by $a$ or $b$, whether the remaining arrows are 
labeled by $a$ or $b$ is also determined. After this, between two adjacent 
vertices, there can be at 
most 4 possibilities: (1) an edge pointing left; (2) an edge pointing right; 
(3) the starting configuration of $a_0$-gadget; (4) the ending configuration of 
$a_0$-gadget. Summarizing the above, there are at most $2\cdot 4^d$ 
indecomposables of dimension $d$ coming from symmetric strings with $e$ acting as 
identity. When $e$ acts as $0$ the counting task is similar, except that in the 
$e$ gadget since the image of $e$ is trivial, those vertices in $e$ do not 
contribute to a dimension. Therefore, summarizing the two cases we have $4^{d+1}$ 
is an upper bound. Of course, due to the 
aforementioned restrictions on the words, and the fact that we need to respect the 
$a_0$-gadget, an arbitrary configuration may not yield a valid word, so $2\cdot 
4^d$ is a very loose bound, but is nonetheless good enough for our purposes. 

Similar considerations yield upper bounds for other types. 

For asymmetric strings, the auxiliary algebra only contributes two 
indecomposable, 
namely the vector space of dimension $1$ with identity map, or with the $0$ map. 
Therefore, taking into account the configurations, a representation of 
dimension $d$ is then determined by a quiver with $d$ vertices arranged on a 
line.\footnote{Here we mean the quiver \emph{after} expansion. Therefore, the 
drawings 
for the 
$a_0$ gadget and the $e$ gadget need to be rotated 90 degrees for the vertices to 
be on a line.} We first have the freedom to set 
the left-most edge to be labeled by $a$ or $b$. After this is fixed, for two 
adjacent vertices, there are $4$ possibilities: (1) an edge pointing left; (2) an 
edge pointing right; 
(3) the starting configuration of $a_0$-gadget; (4) the ending configuration of 
$a_0$-gadget. Therefore $4^{d+1}$ is an upper bound. 

For asymmetric bands, the continuous part is given by Jordan blocks. Since we work 
over $\F_4$, for a fixed dimension $d'$ there are $3$ Jordan blocks with nonzero 
eigenvalues. To count the 
number of indecomposables arising from asymmetric bands of dimension $d$, we shall 
count for each divisor of $d$ separately. This adds a factor of at most $d$. Now 
for a fixed decomposition $d=d'n$, we assume the Jordan blocks are of dimension 
$d'$, and the rest is to count the number of asymmetric bands with $n$ vertices. 
Note that we assume the edge between the first two vertices is labeled with $b$, 
so for each edge whether it is labeled by $a$ or $b$ will be fixed. As before 
there are $4$ possibilities between two adjacent vertices. Thus $3\cdot 4^n$ is an 
upper bound for the decomposition $d=d'n$. Taking into all such decompositions 
$3d\cdot 4^{d}$ is then an upper bound. 

For symmetric bands, it can be done similarly as for symmetric strings. The main  
difference is that the indecomposables are from the four-subspace 
quiver, therefore could possibly contribute a one-parameter family. This can be 
accommodated as in the case of asymmetric bands, therefore giving a $d\cdot 
4^{d+1}$ upper bound. 
\end{proof}

\subsection{A lower bound for wild types} \label{subsec:wild}

We now explain why \emph{every} wild group algebra does not satisfy 
Criterion~\ref{criterion:key}:

\begin{observation}[{J. Rickard \cite{Rickard}}]
Let $\F_p G$ be a group algebra of wild type. Then there are $p^{\Omega(d^2)}$ 
indecomposable $\F_p G$-modules of dimension $d$.
\end{observation}

\begin{proof}
To start with, consider the indecomposable modules of $\F_p\langle x, y\rangle$---the \emph{non-commutative} polynomial ring in two non-commuting variables $x,y$ with coefficients in $\F_p$---of 
dimension $d$ of the following form: fix $A$ to be the single Jordan block of 
size $d$. For any matrix $B$ of size $d$, $(x, y)\to 
(A, B)$ gives an 
indecomposable module of $\F_p\langle x, y\rangle$ of dimension $d$. There are thus
$p^{d^2}$ of such modules. $(A, B)$ and $(A, B')$ are isomorphic, \Iff 
there exists $C\in\GL(d, p)$ \st $CA=AC$, and $CB=B'C$. The number of $C\in 
\GL(d, p)$ \st $CA=AC$ is upper bounded by $p^{d}$ (one can easily compute the 
set of matrices that commute with a single Jordan block), so the number of 
non-isomorphic modules of this form is lower bounded by $p^{d^2 - d}$. 

Let $\F_p G$ be a group algebra of wild type. By definition,\footnote{This follows 
from the precise definition wildness, see {\cite[Sec.~4.4]{Ben98}}.} there exists 
a map 
from $\F_p\langle x, y\rangle$-modules to $\F_p G$-modules, which preserves 
indecomposability and non-isomorphisms, and multiplies the dimension by some 
constant depending only on $G$. Therefore, asymptotically, the number of 
indecomposables 
of $\F_p G$ of dimension $d$ is lower bounded by $p^{c\cdot d^2}$ for some 
constant $c$.
\end{proof}

\section{Algorithm and bounds for cohomology class 
isomorphism}\label{sec:cohomology}

In this section we give the full details of the polynomial-time algorithm for 
\cohiso provided that the coset for \actcomp is 
given (\ie, we solve Problem~\ref{prob:coh_iso} in the tame case). Before we 
begin, we note that the proofs here use cohomology in a black-box fashion that can 
be understood by simple 
pattern-matching, even if the reader is not so familiar with cohomology.

\begin{theorem} \label{thm:coho}
Let $\mathcal{S}$ be an Abelian characteristic subgroup functor. Given two groups 
$G_1, G_2$, and the coset of action compatibilities for the actions $\theta_i$ of 
$G_i/\mathcal{S}(G_i)$ on $\mathcal{S}(G_i)$, one can determine the coset of 
isomorphisms $\Iso(G_1, G_2)$ in time polynomial in $|H^2(G_i/\mathcal{S}(G_i), 
\mathcal{S}(G_i), \theta_i)|$. 
\end{theorem}

\begin{proof}
Let $\alpha$ be an action compatibility, and let $\alpha_1, \dotsc, \alpha_k$ 
generate the group of self-compatibilities for the action associated to $G_2$. By 
applying $\alpha$ to $G_1$, we may assume that $\alpha = 1$. Now we treat each 
$\alpha_i$ as a permutation on $H^2(G/\mathcal{S}(G), \mathcal{S}(G), \theta)$. 
Compute the 2-cohomology classes of the extensions $\mathcal{S}(G_i) 
\hookrightarrow G_i \twoheadrightarrow Q_i$, and now check if they are in the same 
orbit of the permutation group generated by the $\alpha_i$ acting on $H^2$. The 
latter is an instance of the pointwise transporter problem, which can be solved in 
time polynomial in the domain size \cite{seressbook}. One element taking a 
2-cohomology class to the 
other provides an isomorphism, and the stabilizer of the 2-cohomology class gives 
generators of the automorphism group.
\end{proof}

To get Theorem~\ref{thm:main} from the preceding one, we will show that the 
following criterion holds in both finite and tame types. Then 
Theorem~\ref{thm:action} is used to find the coset of action compatibilities, and 
Theorem~\ref{thm:coho} is used to find the coset of group isomorphisms.

In the rest of this section, instead of writing $H^2(Q, \Z_p^k, \theta)$, we 
understand $\theta$ as defining a module $M$ over $\F Q$, and write $H^2(Q, M)$.

\begin{criterion}\label{criterion:key_coh}
The size of $H^2(Q, M)$ is bounded by $\poly(|Q|, |M|)$. Equivalently, the 
dimension of $H^2(Q, M)$ over $\F_p$ is bounded by $O(\dim_{\F_p} M + \log_p 
|Q|)$. 
\end{criterion}

The rest of this section is devoted to showing that in the finite and tame cases 
we in fact get the stronger statement that $\dim H^2(Q, M) \leq O(\dim M)$. We 
start with the case of finite type:

\begin{lemma}[{See \cite[Lemma~3.5]{GKKL}}] \label{lem:GKKL}
Let $\F$ be a field of characteristic $p$, and let $Q$ be a group with a cyclic 
Sylow $p$-subgroup. If $M$ is an indecomposable $\F Q$-module, then for any $j 
\geq 0$, we have $\dim_\F H^j(Q, M) \leq 1$.
\end{lemma}

\begin{proposition} \label{prop:coh_finite}
Let $\F$ be a field of characteristic $p$, and let $Q$ be a group with a cyclic 
Sylow $p$-subgroup. If $M$ is \emph{any} $\F Q$-module, then for any $j \geq 0$, 
we have $\dim_\F H^j(Q, M) \leq \dim_\F M$.
\end{proposition}

\begin{proof}
Write $M = \bigoplus_{i=1}^k M_i$ where the $M_i$ are indecomposable. Since 
$H^j(Q, \bigoplus_i M_i) \cong \bigoplus_i H^j(Q, M_i)$ \cite[p.~34]{Ben98}, we 
get that $\dim_\F H^j(Q, M)$ is at most the number of indecomposable summands of 
$M$, which is at most the dimension of $M$.
\end{proof}

The rest of this section is devoted to showing:

\begin{proposition} \label{prop:coh_tame}
Let $\F$ be a field of characteristic two. If the Sylow $2$-subgroup of a group 
$Q$ 
is dihedral, semi-dihedral, or generalized quaternion, then for any $\F Q$-module 
$M$ we have $\dim_\F H^2(Q, M) \leq 3 \dim_\F M$.
\end{proposition}

The form of this next lemma is from \cite[Lemma~3.8]{GKKL}, but the result is a 
direct consequence of the Lyndon--Hochschild--Serre spectral sequence 
(\cite[p.~337]{maclaneBook} and \cite{holt}) and standard facts about 
low-dimensional cohomology groups (\cite[pp.~354--355]{maclaneBook} and 
\cite[Lemma~2.1]{holt2}). 

\begin{lemma} \label{lem:serre}
Let $N$ be a normal subgroup of $Q$, $\F$ a field, and $M$ an $\F Q$-module. Then 
\[
\dim H^2 (Q, M) \leq \dim H^2 (Q/N, M^N) + \dim H^1 (Q/N, H^1(N, M)) + \dim H^2(N, 
M)^Q,
\]
where $M^N$ denotes the $Q/N$-submodule of $M$ consisting of the $N$-fixed points, 
and $H^2(N, M)^Q$ denotes the $Q$-fixed points in $H^2(N,M)$ (note that here $Q$ 
acts on both $N$ and $M$).
\end{lemma}

As in the case of \actcomp, for cohomology with coefficients in an $\F_p 
Q$-module, we can essentially reduce from $Q$ to its Sylow $p$-subgroup: 

\begin{lemma} \label{lem:sylow}
Let $Q$ be a finite group, $\F$ a field of characteristic $p$, and $M$ an $\F 
Q$-module. If $H$ is a subgroup of $Q$ that contains a Sylow $p$-subgroup of $Q$, 
then there is an $\F H$-module $W$ such that, for all $j \geq 0$, $H^j(H, W) \geq 
H^j(Q, M)$ and $\dim W \leq \dim M$.
\end{lemma}

This follows from standard cohomological results; the proof we give here is from 
\cite[Lemma~3.5]{GKKL}.

\begin{proof}
By \cite[Corollary~3.6.10]{Ben98}, $M$ is projective relative to $H$, thus a theorem of 
Higman  \cite[Proposition~3.6.4]{Ben98} applies to $M$. The latter says that $M$ is a 
direct summand of some module induced from an $\F H$-module $W$. By Frobenius 
reciprocity, $W$ appears in the restriction of $M$ to $H$, so $\dim W \leq \dim 
M$. Finally, Shapiro's Lemma (see, \eg, \cite[Lemma~3.4]{GKKL}) says that 
$H^j(H, W) \cong H^j(Q, W_H^Q)$, and as $M$ is a direct summand of $W_H^Q$, 
$H^j(Q, W_H^Q) \geq H^j(Q, M)$. 
\end{proof}

Finally, we prove Proposition~\ref{prop:coh_tame}, thereby proving 
Theorem~\ref{thm:coho} and Theorem~\ref{thm:main}.

\begin{proof}[Proof of Proposition~\ref{prop:coh_tame}]
By Lemma~\ref{lem:sylow} it suffices to prove the result in the case when $Q$ 
itself is dihedral, semi-dihedral, or generalized quaternion; write $Q = G_m \in 
\{\dihedral{m}, \sdihedral{m}, \gquaternion{m}\}$. All three of the possibilities 
for $G_m$ are extensions of the form $\Z_{2^m} \hookrightarrow G_m 
\twoheadrightarrow \Z_2$. By Lemma~\ref{lem:serre}, we thus have
\begin{eqnarray*}
\dim H^2(G_m, M) & \leq & \dim H^2 (\Z_2, M^{\Z_{2^m}}) + \dim H^1 (\Z_2, 
H^1(\Z_{2^m}, M)) + \dim H^2(\Z_{2^m}, M)^{G_m} \\
 & \leq & \dim H^2 (\Z_2, M^{\Z_{2^m}}) + \dim H^1 (\Z_2, H^1(\Z_{2^m}, M)) + \dim 
 H^2(\Z_{2^m}, M) \\
 & \leq & \dim M^{\Z_{2^m}} + \dim H^1(\Z_{2^m}, M) + \dim M \\
 & \leq & \dim M + \dim M + \dim M
\end{eqnarray*}
Each of the last two lines follows from the preceding line by 
Proposition~\ref{prop:coh_finite}. \end{proof}

Although this already gives us the algorithmic consequence we need, we show how to 
extend Proposition~\ref{prop:coh_tame} to a wider class of groups, in a way that 
may be useful in future work:

\begin{corollary}
Let $G$ be a group with a subnormal series $1 = G_0 \unlhd G_1 \unlhd \dotsb 
\unlhd G_c = G$ where each quotient $G_i/G_{i-1}$ is either simple or cyclic. Let 
$F$ be a field and $M$ an $FG$-module. Then
$\dim_F H^2(G, M) \leq (35/8) c(c+1) \dim_F M$. 

In particular, for the class of groups for which there is such a chain with $c$ 
bounded by $O(1)$, we get that $\dim H^2(G, M) \leq O(\dim M)$. 
\end{corollary}

Note that, in general, we only have the bound $c \leq \log_2 |G|$, which would 
only yield the bound $\dim H^2(G, M) \leq O(\log^2|G| \cdot \dim M)$, whereas 
Criterion~\ref{criterion:key_coh} requires a bound of the form $O(\dim M + \log 
|G|)$. 

\begin{proof}
We apply the same proof inductively, using a few additional facts about the 
cohomology of simple groups. First, Guralnick and Hoffman \cite{GH} showed that 
for simple $G$, any field $F$, and any $FG$-module $M$, $\dim H^1(G, M) \leq (1/2) 
\dim M$. Second, under the same conditions, Guralnick, Kantor, Kassabov, and 
Lubotzky \cite{GKKL} showed that $\dim H^2(G, M) \leq 17.5 \dim M$. Third, 
analogous to Lemma~\ref{lem:serre}, if $G$ is a finite group, $N$ is a normal 
subgroup, $F$ is a field and $M$ is an $FG$-module, then $\dim H^1(G, M) \leq \dim 
H^1(G/N, M^N) + \dim H^1(N, M)^G$. Let $T_i(c)$ be an upper bound on $\dim H^i(G, 
M) / \dim M$ for all $G$ with a chain of length $c$. 

By mimicking the proof of Proposition~\ref{prop:coh_tame}, we get the following 
recurrence for $T_2(c)$:
\begin{eqnarray*}
T_2(c) & \leq & 17.5 + 17.5 T_1(c-1) + T_2(c-1) \\
T_1(c) & \leq & 0.5 + T_1(c-1).
\end{eqnarray*}
Together with the fact that $T_1(1) \leq 0.5$ and $T_2(1) \leq 17.5$, we get that 
$T_1(c) \leq c/2$. Plugging into the bound for $T_2$, we get that $T_2(c) \leq 
\frac{35}{2} ((c-1)/2 + 1) + T_2(c-1)$, and therefore $T_2(c) \leq \frac{35}{4} 
\sum_{i=0}^{c+1} i$, and the result follows.
\end{proof}

\begin{remark}
Guralnick, Kantor, Kassabov, and Lubotzky \cite{GKKL} also showed that for any 
finite group $G$, any field $F$, and any \emph{faithful} $FG$-module $M$---that 
is, the only element of $G$ that acts trivially on $M$ is the identity---then 
$\dim H^2(G, M) \leq 18.5 \dim M$. Together with our results, this suggests that, in these cases, \actcomps may be the only real obstacle to \GpI.
\end{remark}

\section{Conclusion}\label{sec:conclusion}
\subsection{Discussion}\label{subsec:tame_wild}
Generally speaking (if somewhat glibly), there are two overarching reasons an 
instance of an isomorphism problem can be easy (not just group isomorphism): 1) 
there are very few possible isomorphisms to check, or 2) there aren't very many 
isomorphism classes and/or they have an explicit classification. Although this is 
a coarse caricature of reality,\footnote{For example, we recognize that this may not apply to certain algorithms for \GrI.} we believe it provides a useful viewpoint. The 
results of \cite{Sav80,Vik96,Kav07} use the classification of Abelian groups (2); 
the results of \cite{BCQ,BCGQ} roughly fall under (1): The number of isomorphisms 
is only $n^{O(\log \log n)}$, and then they use dynamic programming, an algorithm 
for code equivalence, and results on finite simple groups to reduce this to 
polynomial time; the results of \cite{Gal09,QST11,BQ} fall under (2) in the strong 
sense that they rely on the fact that the number of irreducible representations of 
a group $G$ in characteristic $p$ that doesn't divide $G$ is \emph{finite}, and 
all other representations are direct sums of these; and the results of \cite{LW12} 
use 
an essentially \emph{finite} 
classification of type (2) to reduce to (1) (see \cite{GQ_progress}).
We show that when (2) holds---of which tameness is a general interpretation---isomorphism can be tested in $\P$. 

Because of the universal property of wildness---it is as hard as classifying the 
representations of \emph{any} finite-dimensional algebra---it is widely believed 
that an explicit classification is impossible for wild problems. However, this 
does not rule out structural information, nor does it necessarily rule out 
efficient algorithms to decide when two points are equivalent under a wild 
equivalence relation (for example, as in \cite{CIK97,BL08,IKS10}). However, the 
wild problems that arise in \GpI are frequently ``wilder than wild'' \cite{belitskiiSergeichuk}(analogous to a problem being $\NP$-hard but not in $\NP$), and these problems seem to pose a 
core difficulty for \GpI.

The reasons (1) and (2)---or rather, their absence---also partially explain the 
widely held belief that nilpotent groups of class 2---those $G$ for which $G$ 
modulo its center is Abelian---are the hardest cases of group isomorphism, despite 
the lack of a formal reduction. Option (1) is ruled out, because even for 
$p$-groups of class 2 (nilpotent groups of class 2 and order a power of the prime 
$p$) in which every element is of order $p$, there are roughly $n^{O(\log n)}$ 
possible isomorphisms to check.\footnote{This is essentially because $\Aut(\Z_p^k) 
\cong \GL(k,\F_p)$, which is of size $\sim p^{k^2} = n^{\Theta(\log n)}$.} 
Option (2) is also ruled out, because the $p$-groups of class 2 form a wild 
classification problem \cite{sergeichukPGroups}, and in fact, one that is 
``strictly wilder'' than classifying the representations of finite-dimensional 
algebras \cite{belitskiiSergeichuk}.

These facts, the upper bounds in this paper, and the lower bound on the number of 
indecomposables in wild type, suggest that the border between tame and wild may 
also be the current border between the easy and hard cases of \GpI.

\subsection{Open questions}

\begin{question}
Upgrade Corollary~\ref{cor:main} to groups whose radicals have Abelian Sylow towers, 
that is, drop the requirement that the Sylow subgroups are \emph{elementary} 
Abelian.
\end{question}

Although we believe this is possible, we note that if one tries to use the methods 
of this paper, they must be used ``in a single shot:'' an Abelian Sylow tower can 
always be refined, as in Cannon and Holt \cite{CH03}, to a subnormal series whose 
quotients are elementary Abelian. However, if the Sylow subgroups were not 
themselves elementary Abelian, such as $\Z_{p^2}^d$, then the resulting subnormal 
series will contain more than one factor of the same characteristic, in which case 
proceeding inductively is likely to \emph{appear} to run into wildness. However, 
we believe that it may be possible to extend the structure of tameness, and the 
results that we leveraged here, from $\F_p Q$-modules to $(\Z/p^k \Z) Q$-modules, 
which would be enough to handle Abelian subgroups of characteristic $p$ and 
exponent $p^k$ all at once. (We note that we wouldn't really think of this 
approach as truly handling a wild situation, so much as realizing that a 
particular situation that might seem wild is in fact tame.)

\paragraph{\it Acknowledgment.} We thank G\'abor Ivanyos for pointing out to us 
reference \cite{BO08}. J. A. Grochow is supported by an SFI Omidyar Fellowship 
during this work. Y. Qiao is supported by Australian Research Council DECRA 
DE150100720 during this work.

\bibliographystyle{plain}
\bibliography{ref_new}

\begin{thebibliography}{10}

\bibitem{ABG}
J.~L. Alperin, Richard Brauer, and Daniel Gorenstein.
\newblock Finite groups with quasi-dihedral and wreathed {Sylow}
  {$2$}-subgroups.
\newblock {\em Trans. Amer. Math. Soc.}, 151:1--261, 1970.

\bibitem{Babaisurvey}
L\'{a}szl\'{o} Babai.
\newblock Automorphism groups, isomorphism, reconstruction.
\newblock In R.~L. Graham, M.~Gr\"{o}tschel, and L.~Lov\'{a}sz, editors, {\em
  Handbook of combinatorics (vol. 2)}, pages 1447--1540. MIT Press, Cambridge,
  MA, USA, 1995.

\bibitem{BCGQ}
L{\'a}szl{\'o} Babai, Paolo Codenotti, Joshua~A. Grochow, and Youming Qiao.
\newblock Code equivalence and group isomorphism.
\newblock In {\em Proc. 22nd SODA}, pages 1395--1408, 2011.

\bibitem{BCQ}
L{\'a}szl{\'o} Babai, Paolo Codenotti, and Youming Qiao.
\newblock Polynomial-time isomorphism test for groups with no {Abelian} normal
  subgroups - (extended abstract).
\newblock In {\em ICALP}, pages 51--62, 2012.

\bibitem{BQ}
L\'aszl\'o Babai and Youming Qiao.
\newblock Polynomial-time isomorphism test for groups with {A}belian {S}ylow
  towers.
\newblock In {\em 29th STACS}, pages 453 -- 464. Springer LNCS 6651, 2012.

\bibitem{belitskiiSergeichuk}
Genrich~R. \Belitskii and Vladimir~V. \Sergeichuk.
\newblock Complexity of matrix problems.
\newblock {\em Linear Algebra Appl.}, 361:203--222, 2003.
\newblock Ninth Conference of the International Linear Algebra Society (Haifa,
  2001).

\bibitem{bender}
Helmut Bender.
\newblock Finite groups with dihedral {Sylow} {$2$}-subgroups.
\newblock {\em J. Algebra}, 70(1):216--228, 1981.

\bibitem{Ben98}
D.J. Benson.
\newblock {\em Representations and Cohomology: Volume 1, Basic Representation
  Theory of Finite Groups and Associative Algebras}.
\newblock Cambridge Studies in Advanced Mathematics. Cambridge University
  Press, 1998.

\bibitem{BD82}
V.M. Bondarenko and Yu.A. Drozd.
\newblock Representation type of finite groups.
\newblock {\em Journal of Soviet Mathematics}, 20(6):2515--2528, 1982.

\bibitem{BL08}
Peter~A. Brooksbank and Eugene~M. Luks.
\newblock Testing isomorphism of modules.
\newblock {\em Journal of Algebra}, 320(11):4020 -- 4029, 2008.

\bibitem{BO08}
Peter~A Brooksbank and Eamonn~A. O'Brien.
\newblock Constructing the group preserving a system of forms.
\newblock {\em International Journal of Algebra and Computation},
  18(02):227--241, 2008.

\bibitem{brustle}
Thomas Bruestle.
\newblock Typical examples of tame algebras.
\newblock In {\em Representations of finite dimensional algebras and related
  topics in {L}ie theory and geometry}, volume~40 of {\em Fields Inst.
  Commun.}, pages 27--44. Amer. Math. Soc., Providence, RI, 2004.

\bibitem{CH03}
John~J. Cannon and Derek~F. Holt.
\newblock Automorphism group computation and isomorphism testing in finite
  groups.
\newblock {\em J. Symb. Comput.}, 35:241--267, March 2003.

\bibitem{CIK97}
Alexander~L. Chistov, G{\'a}bor Ivanyos, and Marek Karpinski.
\newblock Polynomial time algorithms for modules over finite dimensional
  algebras.
\newblock In {\em ISSAC}, pages 68--74, 1997.

\bibitem{conradNotes}
Keith Conrad.
\newblock Generalized quaternions.
\newblock
  \url{http://www.math.uconn.edu/~kconrad/blurbs/grouptheory/genquat.pdf},
  2013.

\bibitem{CB89}
W.~W. Crawley-Boevey.
\newblock Functorial filtrations {III}: Semidihedral algebras.
\newblock {\em Journal of the London Mathematical Society}, s2-40(1):31--39,
  1989.

\bibitem{CR66}
C.W. Curtis and I.~Reiner.
\newblock {\em Representation Theory of Finite Groups and Associative
  Algebras}.
\newblock AMS Chelsea Publishing Series. Interscience Publishers, 1966.

\bibitem{Drozd80}
Ju.A. Drozd.
\newblock Tame and wild matrix problems.
\newblock In Vlastimil Dlab and Peter Gabriel, editors, {\em Representation
  Theory II}, volume 832 of {\em Lecture Notes in Mathematics}, pages 242--258.
  Springer Berlin Heidelberg, 1980.

\bibitem{FN}
V.~Felsch and J.~Neub\"user.
\newblock On a programme for the determination of the automorphism group of a
  finite group.
\newblock In Pergamon J.~Leech, editor, {\em Computational Problems in Abstract
  Algebra (Proceedings of a Conference on Computational Problems in Algebra,
  Oxford, 1967)}, pages 59--60, Oxford, 1970.

\bibitem{GP}
I.~M. \Gelfand and V.~A. Ponomarev.
\newblock Problems of linear algebra and classification of quadruples of
  subspaces in a finite-dimensional vector space.
\newblock In {\em Hilbert space operators and operator algebras ({P}roc.
  {I}nternat. {C}onf., {T}ihany, 1970)}, pages 163--237. Colloq. Math. Soc.
  J\'anos Bolyai, 5. North-Holland, Amsterdam, 1972.

\bibitem{gorenstein}
Daniel Gorenstein.
\newblock {\em Finite groups}.
\newblock Chelsea Publishing Co., New York, second edition, 1980.

\bibitem{gorensteinWalterAll}
Daniel Gorenstein and John~H. Walter.
\newblock The characterization of finite groups with dihedral {S}ylow
  {$2$}-subgroups. {I--III}.
\newblock {\em J. Algebra}, 2:85--151, 218--270, 354--393, 1965.

\bibitem{grochowLie}
Joshua~A. Grochow.
\newblock Matrix isomorphism of matrix {L}ie algebras.
\newblock In {\em IEEE Conference on Computational Complexity}, pages 203--213,
  2012.
\newblock Also available as arXiv:1112.2012 and ECCC TR11-168.

\bibitem{GQccc}
Joshua~A. Grochow and Youming Qiao.
\newblock Algorithms for group isomorphism via group extensions and cohomology.
\newblock In {\em IEEE Conference on Computational Complexity (CCC14)}, pages
  110--119, 2014.
\newblock Also available as arXiv:1309.1776 [cs.DS] and ECCC Technical Report
  TR13-123. Submitted for journal publication.

\bibitem{GQ_progress}
Joshua~A. Grochow and Youming Qiao.
\newblock On $p$-group isomorphism and the tame-wild dichotomy.
\newblock In preparation, 2015.

\bibitem{GKKL}
Robert Guralnick, William~M. Kantor, Martin Kassabov, and Alexander Lubotzky.
\newblock Presentations of finite simple groups: profinite and cohomological
  approaches.
\newblock {\em Groups Geom. Dyn.}, 1(4):469--523, 2007.
\newblock Preprint available as arXiv:0711.2817v1 [math.GR].

\bibitem{GH}
Robert~M. Guralnick and Corneliu Hoffman.
\newblock The first cohomology group and generation of simple groups.
\newblock In {\em Groups and geometries ({S}iena, 1996)}, Trends Math., pages
  81--89. Birkh\"auser, Basel, 1998.

\bibitem{Hig54}
D.~G. Higman.
\newblock Indecomposable representations at characteristic $p$.
\newblock {\em Duke Math. J.}, 21(2):377--381, 06 1954.

\bibitem{holt}
D.~F. Holt.
\newblock Exact sequences in cohomology and an application.
\newblock {\em J. Pure Appl. Algebra}, 18(2):143--147, 1980.

\bibitem{holt2}
D.~F. Holt.
\newblock On the second cohomology group of a finite group.
\newblock {\em Proc. London Math. Soc. (3)}, 55(1):22--36, 1987.

\bibitem{IKS10}
G{\'a}bor Ivanyos, Marek Karpinski, and Nitin Saxena.
\newblock Deterministic polynomial time algorithms for matrix completion
  problems.
\newblock {\em SIAM J. Comput.}, 39(8):3736--3751, 2010.

\bibitem{Kav07}
Telikepalli Kavitha.
\newblock Linear time algorithms for {A}belian group isomorphism and related
  problems.
\newblock {\em J. Comput. Syst. Sci.}, 73(6):986--996, 2007.

\bibitem{struct}
Johannes K\"{o}bler, Uwe Sch\"{o}ning, and Jacobo Tor\'{a}n.
\newblock {\em The graph isomorphism problem: its structural complexity}.
\newblock Birkhauser Verlag, Basel, Switzerland, 1993.

\bibitem{Gal09}
Fran\c{c}ois Le~Gall.
\newblock Efficient isomorphism testing for a class of group extensions.
\newblock In {\em Proc. 26th STACS}, pages 625--636, 2009.

\bibitem{LW12}
Mark~L. Lewis and James~B. Wilson.
\newblock Isomorphism in expanding families of indistinguishable groups.
\newblock {\em Groups - Complexity - Cryptology}, 4(1):73–110, 2012.

\bibitem{maclaneBook}
Saunders MacLane.
\newblock {\em Homology}.
\newblock Classics in Mathematics. Springer-Verlag, Berlin, 1995.
\newblock Reprint of the 1975 edition.

\bibitem{Mil78}
Gary~L. Miller.
\newblock On the $n^{\log n}$ isomorphism technique (a preliminary report).
\newblock In {\em Proc. 10th ACM STOC}, pages 51--58, New York, NY, USA, 1978.
  ACM Press.

\bibitem{gct_acm}
Ketan Mulmuley.
\newblock On {$\cc{P}$ vs. $\cc{NP}$} and geometric complexity theory.
\newblock {\em J. ACM}, 58(2):5, 2011.

\bibitem{nazarova}
L.~A. Nazarova.
\newblock Representations of a tetrad.
\newblock {\em Izv. Akad. Nauk SSSR Ser. Mat.}, 31:1361--1378, 1967.

\bibitem{QST11}
Youming Qiao, Jayalal M.~N. Sarma, and Bangsheng Tang.
\newblock On isomorphism testing of groups with normal {H}all subgroups.
\newblock In {\em Proc. 28th STACS}, pages 567--578, 2011.

\bibitem{Rickard}
Jeremy Rickard.
\newblock Answer to: the number of indecomposable modules of finite groups over
  finite fields of a fixed dimension.
\newblock \url{http://mathoverflow.net/a/194773/8012}.

\bibitem{Rosen2}
David Rosenbaum.
\newblock Bidirectional collision detection and faster algorithms for
  isomorphism problems.
\newblock arXiv:1304.3935 [cs.DS], 2013.

\bibitem{Sav80}
Carla Savage.
\newblock An ${O}(n^2)$ algorithm for {A}belian group isomorphism.
\newblock Technical report, North Carolina State University, 1980.

\bibitem{seressbook}
\'Akos Seress.
\newblock {\em Permutation Group Algorithms}.
\newblock Cambridge University Press, 2003.

\bibitem{sergeichukPGroups}
V.~V. \Sergeichuk.
\newblock The classification of metabelian {$p$}-groups.
\newblock In {\em Matrix problems ({R}ussian)}, pages 150--161. Akad. Nauk
  Ukrain. SSR Inst. Mat., Kiev, 1977.

\bibitem{Vik96}
Narayan Vikas.
\newblock An ${O}(n)$ algorithm for {A}belian $p$-group isomorphism and an
  ${O}(n \log n)$ algorithm for abelian group isomorphism.
\newblock {\em J. Comput. Syst. Sci.}, 53(1):1--9, 1996.

\end{thebibliography}

\appendix

\section{Indecomposable modules of semi-dihedral groups}\label{app:SD}

In this appendix, to help make the paper more self-contained, we present the description of indecomposables of the 
semi-dihedral algebra over $\F_4$ as given by Crawley-Boevey \cite{CB89}, with the 
aim of determining an explicit upper bound on the number of indecomposables of a 
fixed dimension. 

Recall that the semi-dihedral 
algebra is $\Lambda_\ell=\F_4\langle a, b\mid a^3=b^2=0, a^2=(ba)^\ell b\rangle$, 
where $\ell=2^{m-1}-1$. Let $\lambda, \mu$ be two nonzero field elements in 
$\F_4$. The indecomposable $\Lambda_\ell$-modules are classified into four types: 
asymmetric strings, symmetric strings, asymmetric bands, and symmetric bands. Each 
type will be associated with a set of configurations, and an auxiliary algebra.

\paragraph{Auxiliary algebras.} We first introduce some algebras and their 
indecomposables. 

\begin{enumerate}
\item $A=\F$: modules over $\F$ are just vector spaces 
over $\F$, and the only indecomposable is thus the one-dimensional vector space $\F$.
\item $A=\F[x]/(q(x))$, $q(x)$ a quadratic polynomial with 
distinct roots: two 
indecomposable modules corresponding to two possible eigenspaces of $x$.
\item For asymmetric bands, $A=\F[x,x^{-1}]$: indecomposables are given by Jordan 
blocks of arbitrary dimension with nonzero eigenvalue.
\item For symmetric bands, $A=\F\langle x, y\rangle/(p(x), q(x))$, where $\F \langle x, y \rangle$ denotes the non-commutative polynomial ring, and $p(x)$ and 
$q(x)$ are quadratic (univariate) polynomials with distinct roots: the indecomposables come from the four-subspace 
quiver with an extra conditions, namely the pair of subspaces as eigenspaces of 
$p(x)$ (resp. $q(x)$) are complementary. For the four-subspace quiver, it is 
well-known that for each dimension vector there is at most one one-parameter 
family \cite{nazarova,GP} (see \cite[Section~3.2]{brustle} for more recent coverage). 
\end{enumerate}

For asymmetric strings, the auxiliary algebra is $\F_4$. For symmetric strings, 
it 
is $\F_4[e]/(e^2=e)$. For asymmetric bands, it is $\F_4[x,x^{-1}]$. For symmetric 
bands, it is $\F_4\langle e, f\rangle/(e^2=e, f^2=f)$. 

\paragraph{Configurations.} To start with, we consider the words in the 
alphabet $\{a_i, b_j\mid i\in\{-(\ell+1), \dots, \ell+1\}, j\in\{-1, +1\}\}$ that 
alternate between $a_i$'s and $b_j$'s. For a letter $c_i$ ($c=a$ or $b$), 
$c_i^{-1}=c_{-i}$. For a word $w=w_1\dots 
w_n$, define $w^{-1}=w_n^{-1}\dots w_1^{-1}$. For two words $w$ and $v$, their 
product is $w\cdot v=wv$, the concatenation of $w$ and $v$. The $k$th power of $w$ 
can then be defined. Furthermore, a partial ordering of words is introduced as 
follows: $w<w'$, if (1) $w=w'c_ix$ with $i>0$ for some word $x$; (2) $w'=wc_ix$ 
with $i<0$ for some word $x$; (3) $w=xc_iy$, $w'=xc_jz$, where $i>j$, and $x, y, 
z$ are words. 

To define strings, we further impose conditions and equivalence relations to the 
above words. The condition is that there should be no subwords of the form 
$b_1a_mb_1$, $a_{m+1}b_1$, $b_1a_{m+1}$, 
or $a_ib_1a_j$ where $i, j>0$. The equivalence relation identifies $w$ with 
$w^{-1}$. If $w=w^{-1}$ then call $w$ a symmetric string; otherwise $w$ is an 
asymmetric string. 

To define bands, we also impose conditions and equivalence relations on the words 
that are of even length, and not powers. The conditions is that, no powers include 
subwords of any of the four types as in the condition for strings. The equivalence 
relation is identifying $w$ with all cyclic rotations of $w$ and $w^{-1}$. 

For each string or band, we first associate a preliminary quiver (a directed 
graph) with edges labeled as follows. Recall that we use $c$ to denote either $a$ 
or $b$. 

\begin{description}
\item[Asymmetric strings] Let $w=w_1\dots w_n$ 
be an asymmetric string. The quiver 
is the graph with $n+1$ vertices $\{v_1, \dots, v_{n+1}\}$, with $n$ edges between 
$v_i$ and  $v_{i+1}$ for $i\in[n]$. The edge $E_i$ between $v_i$ and $v_{i+1}$ is 
directed towards $v_i$ \Iff $w_i=c_j$ with $j > 0$, or $j=0$ and 
$w_{i-1}^{-1}\dots w_1^{-1}>w_{i+1}\dots w_n$. For $w_i=c_j$, $E_i$ is labeled 
with $c_{|j|}$. 
\item[Symmetric strings] Let 
$w=za_0z^{-1}$ be a symmetric string, where 
$z=z_1\dots z_n$. The quiver 
is the graph with $n+1$ vertices $\{v_1, \dots, v_{n+1}\}$, with $n+1$ edges, 
including $n$ edges between 
$v_i$ and  $v_{i+1}$ for $i\in[n]$, and a self-loop at $v_{n+1}$. The edge $E_i$ 
between $v_i$ and $v_{i+1}$ is 
directed towards $v_i$ \Iff $w_i=c_j$ with $j > 0$, or $j=0$ and 
$z_{i-1}^{-1}\dots z_1^{-1}>z_{i+1}\dots z_na_0z$. For $z_i=c_j$, $E_i$ is 
labeled with $c_{|j|}$. The self-loop is labeled with $e$. 
\item[Asymmetric bands] Let 
$w=w_1\dots w_n$ be an asymmetric band. By rotating 
and possibly inverting we assume $w_1=b_1$. The quiver 
is the graph with $n$ vertices $\{v_1, \dots, v_n\}$, with $n$ edges between 
$v_i$ and  $v_{i+1}$ for $i\in[n-1]$, and between $v_n$ and $v_1$. The edge $E_i$
between $v_i$ and $v_{i+1}$ is 
directed towards $v_i$ \Iff $w_i=c_j$ with $j > 0$, or $j=0$ and 
$w_{i-1}^{-1}\dots w_1^{-1}w_n^{-1}\dots w_{i+1}^{-1}> w_{i+1}\dots w_nw_1\dots 
w_{i-1}$. For $i=1$, $E_1$ is labeled with $b=x$. For $i>1$, $w_i=c_j$, $E_i$ is 
labeled with $c_{|j|}$. 
\item[Symmetric bands] Let $w=za_0z^{-1}a_0$ (after a possible rotation) be a 
symmetric band  
for some word $z=z_1\dots z_n$. The quiver 
is the graph with $n+1$ vertices $\{v_1, \dots, v_{n+1}\}$, with $n+2$ edges, 
including $n$ edges between 
$v_i$ and  $v_{i+1}$ for $i\in[n]$, and $2$ self-loops at $v_1$ and $v_{n+1}$, 
respectively. The edge $E_i$ between $v_i$ and $v_{i+1}$ is 
directed towards $v_i$ \Iff $z_i=c_j$ with $j > 0$, or $j=0$ and 
$w_{i-1}^{-1}\dots w_1^{-1}w_n^{-1}\dots w_{i+1}^{-1}> w_{i+1}\dots w_nw_1\dots 
w_{i-1}$. For $z_i=c_j$, $E_i$ is labeled with $c_{|j|}$. The self-loop at $v_1$ 
(resp. $v_{n+1}$) is labeled with $f$ (resp. $e$). 
\end{description}

The preliminary quivers need to augmented with the following three types of gadgets 
associated with (1) $a_i$, $|i|>1$; (2) $a_0$; (3) $e$ (and $f$). That is, when an 
edge between is labeled with $a_i$ (resp., $a_0$, $e$, $f$), it needs to be replaced 
by the following gadgets. For future use, let $a=a_1$ and $b=b_1$. In the 
following diagram, we use $V$, $V_i$, and $I_j$ as labels of the vertices, because 
as seen later, these vertices will be eventually labeled by vector spaces which 
are modules of the auxiliary algebras.

\begin{figure}[htbp]
\begin{center}
\input{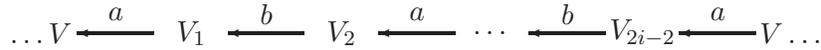}
\caption{The $a_i$ gadget for $|i|>1$. }
\label{fig:ai}
\end{center}
\end{figure}

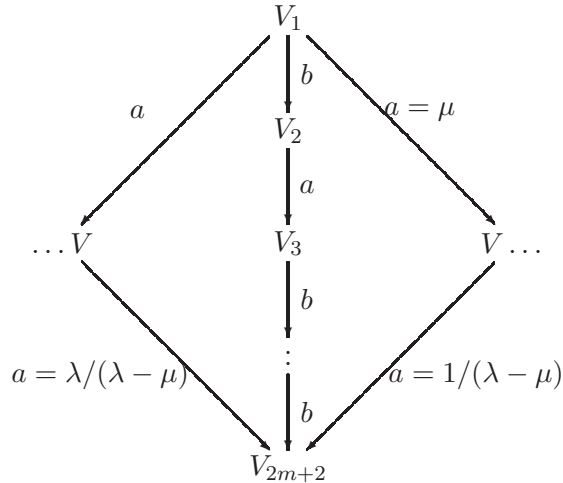
\begin{figure}[htbp]
\begin{center}
\ifx\JPicScale\undefined\def\JPicScale{1}\fi
\unitlength \JPicScale mm
\begin{picture}(60,70)(0,0)
\put(30,70){\makebox(0,0)[cc]{$V_1$}}

\put(30,55){\makebox(0,0)[cc]{$V_2$}}

\put(30,60){\makebox(0,0)[cc]{}}

\put(30,50){\makebox(0,0)[cc]{}}

\put(30,40){\makebox(0,0)[cc]{$V_3$}}

\put(30,25){\makebox(0,0)[cc]{$\vdots$}}

\put(30,10){\makebox(0,0)[cc]{$V_{2m+2}$}}

\put(35,0){\makebox(0,0)[cc]{}}

\put(0,40){\makebox(0,0)[cc]{$\dots V$}}

\put(60,40){\makebox(0,0)[cc]{$V\dots$}}

\linethickness{0.3mm}
\put(30,57.5){\line(0,1){10}}
\put(30,57.5){\vector(0,-1){0.12}}
\linethickness{0.3mm}
\put(30,42.5){\line(0,1){10}}
\put(30,42.5){\vector(0,-1){0.12}}
\linethickness{0.3mm}
\put(30,27.5){\line(0,1){10}}
\put(30,27.5){\vector(0,-1){0.12}}
\linethickness{0.3mm}
\put(30,12.5){\line(0,1){10}}
\put(30,12.5){\vector(0,-1){0.12}}
\linethickness{0.3mm}
\multiput(2.5,42.5)(0.12,0.12){208}{\line(1,0){0.12}}
\put(2.5,42.5){\vector(-1,-1){0.12}}
\linethickness{0.3mm}
\multiput(32.5,67.5)(0.12,-0.12){208}{\line(1,0){0.12}}
\put(57.5,42.5){\vector(1,-1){0.12}}
\linethickness{0.3mm}
\multiput(2.5,37.5)(0.12,-0.12){208}{\line(1,0){0.12}}
\put(27.5,12.5){\vector(1,-1){0.12}}
\linethickness{0.3mm}
\multiput(32.5,12.5)(0.12,0.12){208}{\line(1,0){0.12}}
\put(32.5,12.5){\vector(-1,-1){0.12}}
\put(32.5,62.5){\makebox(0,0)[cc]{$b$}}

\put(32.5,47.5){\makebox(0,0)[cc]{$a$}}

\put(35,32.5){\makebox(0,0)[cc]{}}

\put(32.5,32.5){\makebox(0,0)[cc]{$b$}}

\put(32.5,32.5){\makebox(0,0)[cc]{}}

\put(32.5,17.5){\makebox(0,0)[cc]{$b$}}

\put(55,10){\makebox(0,0)[cc]{}}

\put(10,57.5){\makebox(0,0)[cc]{$a$}}

\put(5,22.5){\makebox(0,0)[cc]{$a=\lambda/(\lambda-\mu)$}}

\put(47.5,57.5){\makebox(0,0)[cc]{$a=\mu$}}

\put(55,22.5){\makebox(0,0)[cc]{$a=1/(\lambda-\mu)$}}

\put(52.5,22.5){\makebox(0,0)[cc]{}}

\end{picture}
\caption{The $a_0$ gadget.  }
\label{fig:a0}
\end{center}
\end{figure}

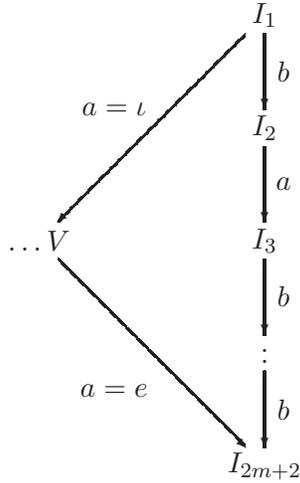
\begin{figure}[htbp]
\begin{center}
\ifx\JPicScale\undefined\def\JPicScale{1}\fi
\unitlength \JPicScale mm
\begin{picture}(32.5,70)(0,0)
\put(30,70){\makebox(0,0)[cc]{$I_1$}}

\put(30,55){\makebox(0,0)[cc]{$I_2$}}

\put(30,40){\makebox(0,0)[cc]{$I_3$}}

\put(30,25){\makebox(0,0)[cc]{$\vdots$}}

\put(30,10){\makebox(0,0)[cc]{$I_{2m+2}$}}

\put(0,40){\makebox(0,0)[cc]{$\dots V$}}

\linethickness{0.3mm}
\put(30,57.5){\line(0,1){10}}
\put(30,57.5){\vector(0,-1){0.12}}
\linethickness{0.3mm}
\put(30,42.5){\line(0,1){10}}
\put(30,42.5){\vector(0,-1){0.12}}
\linethickness{0.3mm}
\put(30,27.5){\line(0,1){10}}
\put(30,27.5){\vector(0,-1){0.12}}
\linethickness{0.3mm}
\put(30,12.5){\line(0,1){10}}
\put(30,12.5){\vector(0,-1){0.12}}
\linethickness{0.3mm}
\multiput(2.5,42.5)(0.12,0.12){208}{\line(1,0){0.12}}
\put(2.5,42.5){\vector(-1,-1){0.12}}
\linethickness{0.3mm}
\multiput(2.5,37.5)(0.12,-0.12){208}{\line(1,0){0.12}}
\put(27.5,12.5){\vector(1,-1){0.12}}
\put(10,57.5){\makebox(0,0)[cc]{$a=\iota$}}

\put(10,20){\makebox(0,0)[cc]{$a=e$}}

\put(32.5,62.5){\makebox(0,0)[cc]{$b$}}

\put(32.5,47.5){\makebox(0,0)[cc]{$a$}}

\put(32.5,32.5){\makebox(0,0)[cc]{$b$}}

\put(32.5,17.5){\makebox(0,0)[cc]{$b$}}

\end{picture}
\caption{The $e$ (resp. $f$) gadget. $I_j=\im e$ (resp. $\im f$), and $\iota$ 
denotes the inclusion of $I_i$ 
in $V$. }
\label{fig:e}
\end{center}
\end{figure}

Now we have all the ingredients to describe the representations of semi-dihedral 
algebras. For each configuration, we represent it using the preliminary 
quiver, and expand the preliminary quiver $Q'$ using the gadgets above to get the 
final quiver $Q$. Note that after 
expansion, the possible edge labels in the final quiver $Q$ are: $a$, $b$, $b=x$, 
or $a=y$, where 
$y\in\{\iota, e, f, \mu, \lambda/(\lambda-\mu), 1/(\lambda-\mu)\}$. Now take an 
indecomposable $V$ from the corresponding auxiliary algebra. Then 
a representation of $\Lambda_\ell$ can be formed as follows. Let $s$ be the number 
of vertices in $Q$. Then the underlying space $U$ is a direct sum of $s$ copies of 
$V$. The linear map corresponding to $b$ is specified by interpreting the label 
$b$ as the identity map between the two copies, and the label $b=x$ as the linear 
map associated with $x$, and otherwise $0$. The linear map 
corresponding to $a$ is specified by interpreting label $a$ as identify map, 
$a=\iota$ as the inclusion map, $e$ (and $f$) as the idempotent linear map, and 
$\mu$, $\lambda/(\lambda-\mu)$, $1/(\lambda-\mu)$ as the scalar map. 

Crawley-Boevey proved that these are all the indecomposables of $\Lambda_\ell$, 
and if two such indecomposables differ on either the continuous part or the 
discrete part, they are non-isomorphic.

\end{document}